\documentclass[conference]{IEEEtran}
\IEEEoverridecommandlockouts
\usepackage{cite}
\usepackage{amsmath,amssymb,amsfonts}
\usepackage{amsthm}
\newtheorem{definition}{Definition}
\newtheorem{theorem}{Theorem}

\newtheorem{assumption}{Assumption}

\DeclareMathOperator{\Tr}{Tr}

\DeclareMathOperator*{\argmin}{arg\,min}

\usepackage{hyperref}

\usepackage{algorithmic}
\usepackage[ruled,vlined,algo2e,linesnumbered]{algorithm2e}
\usepackage{setspace}
\usepackage{graphicx}
\usepackage{subfigure}
\usepackage{textcomp}
\usepackage{pifont}
\usepackage{xcolor}

\def\BibTeX{{\rm B\kern-.05em{\sc i\kern-.025em b}\kern-.08em
    T\kern-.1667em\lower.7ex\hbox{E}\kern-.125emX}}
\begin{document}

\title{LAWS: Look Around and Warm-Start Natural Gradient Descent for Quantum Neural Networks\\
% {\footnotesize \textsuperscript{*}Note: Sub-titles are not captured in Xplore and
% should not be used}
% \thanks{Identify applicable funding agency here. If none, delete this.}
}

\author{
\IEEEauthorblockN{Zeyi Tao, Jindi Wu, Qi Xia and Qun Li}
\IEEEauthorblockA{\textit{Department of Computer Science} \\
\textit{William \& Mary}\\
Williamsburg, VA, USA \\
\{ztao, jwu21, qxia01, liqun\}@cs.wm.edu}
% \and
% \IEEEauthorblockN{2\textsuperscript{nd} Given Name Surname}
% \IEEEauthorblockA{\textit{dept. name of organization (of Aff.)} \\
% \textit{name of organization (of Aff.)}\\
% City, Country \\
% email address or ORCID}
% \and
% \IEEEauthorblockN{3\textsuperscript{rd} Given Name Surname}
% \IEEEauthorblockA{\textit{dept. name of organization (of Aff.)} \\
% \textit{name of organization (of Aff.)}\\
% City, Country \\
% email address or ORCID}
% \and
% \IEEEauthorblockN{4\textsuperscript{th} Given Name Surname}
% \IEEEauthorblockA{\textit{dept. name of organization (of Aff.)} \\
% \textit{name of organization (of Aff.)}\\
% City, Country \\
% email address or ORCID}
% \and
% \IEEEauthorblockN{5\textsuperscript{th} Given Name Surname}
% \IEEEauthorblockA{\textit{dept. name of organization (of Aff.)} \\
% \textit{name of organization (of Aff.)}\\
% City, Country \\
% email address or ORCID}
% \and
% \IEEEauthorblockN{6\textsuperscript{th} Given Name Surname}
% \IEEEauthorblockA{\textit{dept. name of organization (of Aff.)} \\
% \textit{name of organization (of Aff.)}\\
% City, Country \\
% email address or ORCID}
}

\maketitle

\begin{abstract}
Variational quantum algorithms (VQAs) have recently received significant attention from the research community due to their promising performance in Noisy Intermediate-Scale Quantum computers (NISQ).
However, VQAs run on parameterized quantum circuits (PQC) with randomly initialized parameters are characterized by barren plateaus (BP) where the gradient vanishes exponentially in the number of qubits.
In this paper, we first review quantum natural gradient (QNG), which is one of the most popular algorithms used in VQA, from the classical first-order optimization point of view.
Then, we proposed a \underline{L}ook \underline{A}round \underline{W}arm-\underline{S}tart QNG (LAWS) algorithm to mitigate the widespread existing BP issues.
LAWS is a combinatorial optimization strategy taking advantage of model parameter initialization and fast convergence of QNG.
LAWS repeatedly reinitializes parameter search space for the next iteration parameter update.
The reinitialized parameter search space is carefully chosen by sampling the gradient close to the current optimal.
Moreover, we present a unified framework (WS-SGD) for integrating parameter initialization techniques into the optimizer.
We provide the convergence proof of the proposed framework for both convex and non-convex objective functions based on Polyak-Lojasiewicz (PL) condition.
Our experiment results show that the proposed algorithm could mitigate the BP and have better generalization ability in quantum classification problems.
\end{abstract}

\section{Introduction}\label{sec:introduction}
The emergence of the Noisy Intermediate-Scale Quantum (NISQ)~\cite{preskill2018quantum} technology has demonstrated its enormous potential in number factorization~\cite{shor1994algorithms}, quantum system simulation~\cite{lloyd1996universal}, or solving linear systems of equations~\cite{harrow2009quantum}.
Current state-of-the-art NISQ with 50-100 qubits may be able to perform tasks that outperform the capabilities of today’s classical computers~\cite{arute2019quantum}.
However, NISQ is limited by connectivity, qubit count, and gate fidelity, preventing the use of quantum error correction and making many quantum algorithms impractical~\cite{cerezo2021variational}.
 To this end, one of the most promising computational models for using near-term quantum computers is proposed so-called Variational Quantum Algorithms (VQAs)~\cite{mcclean2016theory}.

The VQA is a quantum-classical hybrid algorithm.
In VQA, a task of interest is prepared and evaluated via a parameterized quantum circuit (PQC)  on a quantum computer, with variationally updating the parameters by a classical optimizer to find the optimum of some measurable cost function.
The applications of VQA include the Variational Quantum Eigensolver (VQE)~\cite{peruzzo2014variational}, Quantum Approximate Optimization Algorithm (QAOA)~\cite{farhi2014quantum}, and Quantum Neural Networks (QNNs)~\cite{farhi2018classification}.
The success of VQA is due to 1) VQA allows task-oriented programming making the design of quantum algorithms efficient~\cite{cerezo2021cost}; 2) compared to the classical neural network, the expressibility of QNN is more significant even with shallow quantum circuits~\cite{tilly2021variational};
This low complexity in QNN mitigates the NISQ limitations.

Although it has been shown that the optimization task for minimizing the VQA cost function is, in general, an NP-hard problem~\cite{bittel2021training}, the effectiveness and efficiency of gradient-based optimizers are still charming.
Many employ gradient-based optimizers as a backbone in VQA.
For example, one uses gradient descent to reach the ground-state energy under a Hamiltonian in VQE study.
Alternatively, for some variational classifiers, one uses stochastic gradient descent (SGD) to find the optimal PQC model.
As a result, the performance of QNNs heavily depends on the power of such a classical optimizer.
In particular, the quantum natural gradient (QNG)~\cite{stokes2020quantum} has drawn much research attentions~\cite{yamamoto2019natural,koczor2019quantum} due to its extraordinary ability to discover the parameter space's geometric structure.
Further, QNG has been proved in Ref.~\cite{carleo2017solving} that the VQE associated with QNG is equivalent to the imaginary time evolution~\cite{mcardle2018variational} when the quantum Fubini-Study metric is applied to measure the geometric structure, making it more widespread.

However, a recently discovered phenomenon, so-called barren plateaus (BP)~\cite{mcclean2018barren}, where gradients of the cost functions vanish exponentially with the size of the system, dramatically limits the application of QNNs to practical problems.
BP prevents PQC's parameter update from gradient changes when using gradient-based optimizers.
To acquire the gradient information, exponential resources might be used for sampling errors in quantum measurements.

To address the BP issue, gradient rescaling~\cite{suzuki2021normalized,haug2021optimal}, PQC's parameter initialization~\cite{grant2019initialization,liu2021parameter}, and gradient-free optimizations~\cite{franken2020gradient} have been extensively studied.
Our work is also motivated by addressing the BP issue.
In this paper, we first review the gradient-based method, particularly QNG, in the view of mirror descent~\cite{Nemirovskii1983}. 
Then, we proposed a look around warm-start QNG (LAWS) algorithm as a primary instrument to mitigate the BP issue.
The proposed algorithm is based on two observations:
First, it has been reported in Ref.~\cite{wierichs2020avoiding} that the QNG can consistently find a global optimum and requires significantly fewer epochs than other optimizers.
This outperformance holds even for large system sizes (40 qubits), indicating that using QNG to solve the QNN problem is suitable.
Second, the success of applying parameter initialization in QNN reported in Ref.~\cite{grant2019initialization} demonstrates a potential direction for mitigating the BP issue, where it withstand the possible failure of using the gradient-based~\cite{cerezo2020impact} or gradient-free~\cite{arrasmith2021effect} algorithm.
In addition, many classical machine (deep) learning models benefit from parameter initialization strategies and gain performance improvement~\cite{sutskever2013importance}.
Based on the above, the intuition behind LAWS is that we repeatedly reinitialize the PQC's parameter while in training. 
We call this reinitialization during the training as warm-start.
Essentially, we perform parameter initialization for every current optimum until some criteria meet (i.e., the value of the cost function is minimized).
In this way, the fast convergence speed of QNG is adopted, and the BP could be mitigated via multiple parameter reinitializations.
However, designing such an algorithm is non-trivial and challenging. 
The question, for example, of how to perform reinitialization while in training and how to reinitialization should be carefully treated.
More motivations, discussions, and implementation details are present in section~\ref{sec:main}.

In summary, the contributions of this paper are following:
\begin{itemize}
    \item First, we propose a new derivation of QNG by using a classical first-order optimization scheme known as mirror descent.
    \item Second, we proposed a new algorithm named LAWS for solving VQE and QNN in general. 
    Our experiment results show that the proposed algorithm could mitigate the BP issue and have better generalization ability in quantum classification problems.
    \item Third, based on LAWS, we propose a unified framework WSSGD for the warm-start gradient descent algorithm that is easy to implement and compatible with the most current quantum learning libraries.
    \item Last, as a complementary part, we provide the convergence proof of the proposed framework for both convex and non-convex objective functions.
\end{itemize}

% The presence of BP becomes one of the major bottlenecks in optimizing PQC.
% Notably, this does not preclude VQA, allowing for efficient gradient-based optimization.
% In the earlier sections, we discuss three mainstream techniques to mitigate BP.

% In this work, we focus on the initialization strategy not only because it has been proved helpful in VQA but also the parameter initialization has been considered a promising method in classical machine learning for a long time.
% Many classical machine (deep) learning models benefit from parameter initialization strategies and gain performance improvement.
% Similarly, in PQC, parameter initialization has been reported to have better performance than simply using the raw model, which may mitigate the BP issue.
% For example,~\cite{grant2019initialization} randomly initializes some of the parameters and keeps the remaining parameters in the identity blocks.
% In the same line of research,~\cite{liu2021parameter} uses a well-trained PQC block to initialize a deeper PQC by stacking the blocks and performing transfer learning.
% We say the methods mentioned above are one-shot initialization because they only initialize PQC once.
% However, the study of finding a one-shot initialization strategy for PQC has less attention and remains open.
% In addition, our empirical results show that the one-shot initialization may also suffer BP.
% On the other hand, parameter initialization heavily affects the performance of gradient-based optimization in terms of convergence speed, model performance, and generalization ability.

\section{Related Work}\label{sec:related}
In this section, we first introduce some notation used in the paper, then we summarize the recent studies in optimizing QPC and the presence of BP issue, which has been considered a major limitation in VQA.
\subsection{Notations}
For $\theta, \mu \in R^d$, let $\sqrt{\theta}, \theta\odot \mu$, and $\theta/\mu$ denote the element-wise square root, multiplication, and division of the vectors.
The $\Vert \theta\Vert_2^2$ is $l_2$-norm.
We denote $\theta_{k}^t$ for parameter $\theta$ at $t$-th iteration $k$-th step. 

\subsection{Optimization in VQA}
Hybrid quantum-classical optimization performed via parameterized quantum circuits is a promsing approach for various new emerging quantum based application.
A classical optimization scheme is utilized to update the parameters of such hybrid quantum-classical model.
Two types of optimization are often used: gradient-based methods and gradient-free methods.

\textit{Gradient-based Optimization} has been widely adopted in solving QPC such as stochastic gradient descent (SGD)~\cite{harrow2021low}.
SGD replaces the exact partial derivative with an unbiased gradient estimator at each optimization step.
SGD is a promising method for almost all large-scale machine learning models, where it has been found to be efficient for gradient evaluation, fast convergence speed, etc.
Properly choosing the step size of SGD is essential.
In a recent study~\cite{suzuki2021normalized}, the Vanilla SGD has been replaced with many modern adaptive learning rate methods such as AdaGrad~\cite{duchi2011adaptive}, Adam~\cite{kingma2014adam}, and AdaBelief~\cite{zhuang2020adabelief} to achieve even faster convergence speed and better performance.

The natural gradient~\cite{amari1998natural} (NG) automatically chooses gradient step size and moves in the steepest descent direction with respect to the Fisher information.
The pioneering work~\cite{stokes2020quantum} propose QNG as part of a general-purpose optimization framework for variational quantum algorithms.
Later,~\cite{yamamoto2019natural} demonstrates some simple case studies for QNG via variational quantum eigensolver to reveal how the natural gradient optimizer uses the geometric property to change and improve the ordinary gradient method.
It has been reported~\cite{koczor2019quantum,wierichs2020avoiding} QNG could effectively avoid the local optimal and be stable of performance on all considered system sizes.
QNG's computation is expensive; hence it becomes an obstacle for applying in both classical learning and VQA.
A recent technical note~\cite{jones2020efficient} presents a time-efficient QNG method~\cite{soori2021tengrad} to compute the inverse of the quantum Fisher information matrix.

\textit{Gradient-free Optimization} such as the Nelder-Mead algorithm~\cite{nelder1965simplex}, Powell algorithm~\cite{powell1964efficient}, and Constrained Optimization BY Linear Approximations is also welcome for optimizing PQC models.
~\cite{franken2020gradient} proposes the first gradient-free quantum optimization for NISQ device as a detour solution for BP.
However,~\cite{arrasmith2021effect} reports that gradient-free optimizers do not solve the BP problem due to the exponentially large resources demand.

\subsection{Mitigating the Barren Plateau}\label{subsec:bp}
The barren plateau (BP) phenomenon in the cost function landscape was originally discovered in~\cite{mcclean2018barren} where it was shown that deep (unstructured) parametrized quantum circuits exhibit BPs when randomly initialized.
When a given cost function exhibits a BP, the magnitude of its partial derivatives will be, on average, exponentially vanishing with the system size~\cite{cerezo2021variational}.
Thus, BP has been recognized as a well-known bottleneck in VQA~\cite{cerezo2020impact}, especially when optimizing the QNNs using the gradient-based method.
~\cite{cerezo2020impact} demonstrates that even using high-order derivative information such as the Hessian, the exponential scaling associated with BP still exists.

Many works have been studied to mitigate BP, and they can be roughly categorized into two directions.
The first type of approach uses problem-inspired ansatzes because problem-agnostic ansatzes, such as deep hardware efficient ansatzes, could exhibit barren plateaus due to their high expressibility~\cite{cerezo2021cost,holmes2022connecting}.
The approach, for example~\cite{cerezo2021cost}, relaxes search space during the optimization to a smaller space that contains the solution to the problem or that at least contains a good approximation to the solution while maintaining a low expressibility.

Another line of study focus on QNN initialization.
Parameter initialization has been proved to be helpful in classical machine learning.
In~\cite{grant2019initialization}, the proposed method uses the identity block strategy to limit the effective depth of the circuits used to calculate the first parameter update to avoid the QNN being stuck in a barren plateau at the start of training. 
Further,~\cite{liu2021parameter} proposes a parameter initialization strategy that transfers the small pre-trained layer blocks to the target model stacking by multiple identical basic blocks.
This idea is based on transfer learning.
The empirical results show that the gradient norm's variance is scaled and prove it is effective for mitigating BP where the trainability of QNNs is improved.

The method~\cite{haug2021optimal} uses adaptive learning rates induced from Gaussian kernels for the gradient method to avoid the gradient vanishing.
Work~\cite{larocca2021diagnosing} analyzes the existence of barren plateaus in various ansatzes, and sheds light on the role of the different initial states causing the presence or absence of barren plateaus.
Then they provide an efficient framework for trainability-aware ansatz design strategies.
\section{Background}\label{sec:preliminary}
\subsection{Variational Quantum Algorithms}
The VQA is a quantum-classical hybrid algorithm that enables noisy quantum devices to work with the help of classical computers. 
The VQA runs a parameterized quantum circuit for ansatz preparation and expectation value measurement on a quantum computer. 
Meanwhile, VQA variationally updates the parameters by a classical optimizer to find a global minimum of the objective function.
Generally, the PQC with unitary $U (\theta)$ having the form
%Generally, the ansatz is generated from $| \psi ( \theta}) \rangle = U ( \theta})|0\rangle$, with the PQC unitary $U ( \theta})$ having the form
\begin{equation}\label{eq:unitary}
    U( \theta) = \prod_{l=1}^{L} U_l( \theta_l)
\end{equation}
where
\begin{equation}
    U_l( \theta_l) = \prod_{m=1}^{M}e^{-i\theta_{l,m} V_m / 2} W_{l,m}
\end{equation}
Here, the $l$ indicates the layer where $ \theta_l = (\theta_{l, 1},\cdots, \theta_{l, M})$ contains the parameter of such layer and $ \theta = \{ \theta_l \}_{l=1}^{L}$. 
$V_m$ is Hermition operator that generates the unitary in the ansatz.
In addition, $W_{l,m}$ is the unparameterized quantum gate.
One of the widely used ansatzes in quantum chemistry, optimization, and quantum simulation is hamiltonian variational ansatz (HVA) which aims to prepare a trial ground state for a given Hamiltonian $H = \sum V_m$ (here $V_m$ are Hermitian operators, usual Pauli strings) by Trotterizing an adiabatic state preparation process~\cite{wecker2015progress}. The unitary of HVA, therefore, is given by $U(\theta) = \prod_{l}\big(\prod_{m}\text{exp}(-i\theta_{l,m}V_{l,m})\big)$ in the form of Euqation~\ref{eq:unitary}.
Without loss of generality, the cost of VQA can be expressed as
\begin{equation}
    \mathcal C( \theta) = f\big ( \{\rho_k\}, \{\mathcal O_k\}, U( \theta) \big)
\end{equation}
where $f$ is some function, $\{\rho_k\}$ are input states, $\{\mathcal O_k\}$ is a set of observables, and $U(\theta)$ is parameterized unitary defined in Equation~\ref{eq:unitary}.
To explicitly define function $f$, one could have cost in the form
\begin{equation}
    \mathcal C( \theta) = \sum_{k=1}^{K} a_k \Tr\big [ U( \theta) \rho_k U^{\dagger}( \theta) \mathcal O_k \big ]
\end{equation}
where $a_k$ is the coefficients of the linear combination of expectation values. 
The subscript $k$ indicates the fact that one could work with different functions for each state in the training set.
Particularly, in variational quantum eigensolver, one chooses $\mathcal O = H$, where $H$ is the physical Hamiltonian and lets training set with size $K = 1$. The cost function becomes the expectation value of the Hamiltonian $H$
\begin{equation}\label{eq:costvqe}
    \mathcal C( \theta) = \Tr\big [ U( \theta) \rho_0 U^{\dagger}( \theta) H \big ]
\end{equation}
where $\rho_0 = |  0 \rangle \langle  0 |$ is the initial (pure) state.
In the above Equation~\ref{eq:costvqe}, let $| \psi ( \theta) \rangle = U ( \theta)| 0\rangle$.
For simplicity, we write $| \psi ( \theta) \rangle$ as $| \psi \rangle $ henceforth.
Given a $N$-dimensional complex Hilbert space $\mathbb C^{N}$, we define a projector $P_{\psi}$ as $| \psi \rangle \langle \psi | \in \mathbb C \mathbb P^{N-1}$.
Now we consider the following optimization problem
\begin{equation}\label{eq:generalcost}
    \min_{ \theta} \mathcal C( \theta)
\end{equation}
where
\begin{equation}\label{eq:vqe}
    \mathcal C( \theta) = \Tr(P_{\psi} H) = \langle \psi | H | \psi \rangle
\end{equation}
Note that $\psi$ is normalized since $U( \theta)$ is unitary.

\subsection{First-order Optimization}\label{subsec:firstorder}
Globally optimizing the objective function $\mathcal C( \theta)$ is impractical due to
the nonconvexity.
%It is likely impractical to minimize this nonconvex objective function $\mathcal C( \theta})$ globally.
To this end, practitioners search for local optima by solving the following dynamical system
\begin{equation}\label{eq:basic}
    \theta_{t+1} = \argmin_{\theta} \big \{ \langle \theta, \nabla \mathcal C(\theta_t) \rangle + \frac{1}{2\eta} \Vert \theta - \theta_t \Vert_2^2 \big \}
\end{equation}
which is equivalent to the gradient descent in the form $\theta_{t+1} = \theta_t - \eta \nabla \mathcal C(\theta_t)$. 
Notice, the stochastic gradient descent (SGD) is obtained when $g_t = \nabla \mathcal C(\theta_t, \xi)$ where $\xi$ is a sample drawn from dataset $\mathcal D$ such that $\mathbb E [g_t] = \nabla \mathcal C(\theta_t)$ is a unbaised estimator of $ \nabla \mathcal C(\theta_t)$ and
\begin{equation}\label{eq:basicsgd}
    \theta_{t+1} = \theta_t - \eta g_t
\end{equation}
Optimization problem~\ref{eq:basic} or~\ref{eq:basicsgd} is well-suited to assumptions regarding the objective function $\mathcal C$ which involve the Euclidean norm.
Th intuition behind optimization in Equation~\ref{eq:basic} is objective function $\mathcal C$ is replaced by its linearization at $\theta_t$ plus a Euclidean distance term $\frac{1}{2\eta}\Vert \theta - \theta_t \Vert_2^2$, which prevents the next iterate $\theta_{t+1}$ from being too far from $\theta_t$.

Instead of using Vanilla SGD above, recent studies tackle the optimization problem~\ref{eq:vqe} by using natural gradient descent, where we update the parameter as
\begin{equation}\label{eq:ngd}
    \theta_{t+1} = \theta_t - \eta F(\theta)^{-1}g_t
\end{equation}
Here, $F(\theta) = \Re[G(\theta)]$ is Fubini-Study metric tensor a $P\times P$ matrix recently identified as the (classical) Fisher information matrix.
We define quantum geometric tensor $G(\theta)$ as
\begin{equation}
    G_{i,j} = \Big\langle\frac{\partial \psi}{\partial \theta_i}, \frac{\partial \psi}{\partial \theta_j}\Big\rangle -  \Big\langle\frac{\partial \psi}{\partial \theta_i}, \psi \Big\rangle \Big\langle \psi, \frac{\partial \psi}{\partial \theta_j}\Big\rangle
\end{equation}
Seminar work~\cite{stokes2020quantum} demonstrates the block-wise Fubini-Study metric tensor can be evaluated in terms of quantum expectation values of Hermitian observables which is thus experimentally realizable.
To evaluate quantum geometric tensor,~\cite{van2021measurement} reports that it requires total $\mathcal O(P^2)$ rounds of sampling.
Each round involves a repeated but parallelizable evaluation of an ansatz circuit with $\mathcal O(P)$ gates.
When the global phase of the ansatz state $|\psi\rangle$ is independent of the parameters, the second term of $G_{i,j} $ vanish and $G(\theta)$ becomes equivalent to the coefficient matrix in imaginary time evolution~\cite{mcardle2018variational,yuan2019theory}.
The advantage of using QNG in VQE is that it can easily measure the  distinguishability of the objective function in non-Euclidean parameter space.
% An objective function $f(\cdot, \cdot)$ is said to be indistinguishable if given two parameters $\theta_1, \theta_2$, there might be a point $\theta_1 = \theta$ such that $f(\theta, \theta_2)$ takes the same value for all $\theta_2$.
% Therefore, any two different $\theta_2$ in this case cannot be distinguished.
% In the VQE setting, the metric is defined on the indistinguishability of the function $f(\theta) = \langle \psi(\theta)|H|\psi(\theta) \rangle$ where we measure the distance in the space of pure quantum states $|\psi(\theta_1)\rangle$ and $|\psi(\theta_2)\rangle$.
% To this purpose, Fubini-Study distance in~\ref{eq:ngd} is used for the parameter in quantum space.
% The derivation of the quantum natural gradient in~\cite{yamamoto2019natural} uses a similar distinguishable analysis above.
% Later, in this work, we demonstrate QNG can be derived from probabilistic distance-based analysis.

\subsection{Barren Plateau Problem}
The gradient ($\nabla C(\theta)$) or stochastic gradient ($g$) plays an essential role in the parameter optimization process via the gradient-based method.
However, it has been shown recently that the cost function exhibits a BP where the cost function gradient vanishes exponentially with the system size (the number of qubits).
This phenomenon is also identified as cost concentration in~\cite{arrasmith2021equivalence}.
Here, we consider here the following generic definition of a barren plateau without loss of generality.
\begin{definition}
(Barren Plateau). Consider the cost function $\mathcal C(\theta)$ defined in Eq~\ref{eq:generalcost}. This cost exhibits a barren plateau if, for all $\theta_i \in  \theta$, the expectation value of partial derivative $\partial_i\mathcal C_i(\theta) = \partial\mathcal C(\theta) / \partial \theta_i$ respect to the cost function is zero i.e., $\mathbb E[\partial_i\mathcal C_i(\theta)] = 0$.
The variance of the above partial derivative vanishes exponentially with the number of qubits, i.e,
\begin{equation}
    \text{Var}_{\theta}[\partial_i\mathcal C_i(\theta)] \in \mathcal O(p^{-n})
\end{equation}
for some $p > 1$.
\end{definition}
Notice, we have the following conclusion by using Chebyshev’s inequality
\begin{equation}
    P(|\partial_i\mathcal C_i(\theta)| \geq c) \leq \frac{\text{Var}_{\theta}[\partial_i\mathcal C_i(\theta)]}{c^2} 
\end{equation}
for some constant $c$.
The definition and above inequality tells that the probability of finding a $\partial_i\mathcal C_i(\theta)$ that is larger than $c$ decreases exponentially when the variance of the partial derivative establishes an exponential decay.
The presence of BPs exists in both deep unstructured PQC with randomly initialized parameters~\cite{mcclean2018barren} and QNNs~\cite{grant2019initialization}.
Ref.~\cite{grant2019initialization} theoretically analysis the BP based on the fact that when ansatzes become unitary 2-designs~\cite{harrow2009random}, the expected number of samples required to estimate $\partial C(\theta)$ is exponential in the system size which often refers the number of qubits $n$.
A circuit forms unitary 2-design means that 
the distribution matches up to the second moment that of the uniform Haar distribution of unitaries, which will be used in the analysis of $\partial_i C(\theta)$).
For example the partial derivative of $\theta_k$ in $l$-th layer respect to cost function in~\ref{eq:vqe} can be explicitly described as
\begin{equation}
    \partial_k \mathcal C(\theta) = i \langle 0| U_{-}^{\dagger} [V_k U_{+}^{\dagger}H U_{+}] U_{-}| 0 \rangle
\end{equation}
where $U_{-} = \prod_{l=k-1}^{1} U_{l}(\theta_l)$ and $U_{+} = \prod_{l=L}^{k} U_{l}(\theta_l)$.
BP is fatal in gradient-based optimization because it might halt the parameter update and quickly converge to some sub-optimal solution.
In other words, one needs exponential resources to sample the gradient. 
Therefore, optimizing parameters in the BP region with the gradient-based method becomes hard.
%In this work, we propose a warm-start strategy that might be able to mitigate barren plateau.

\section{Main Results}~\label{sec:main}
\begin{figure}[t]
\centering
\includegraphics[width=.75\columnwidth]{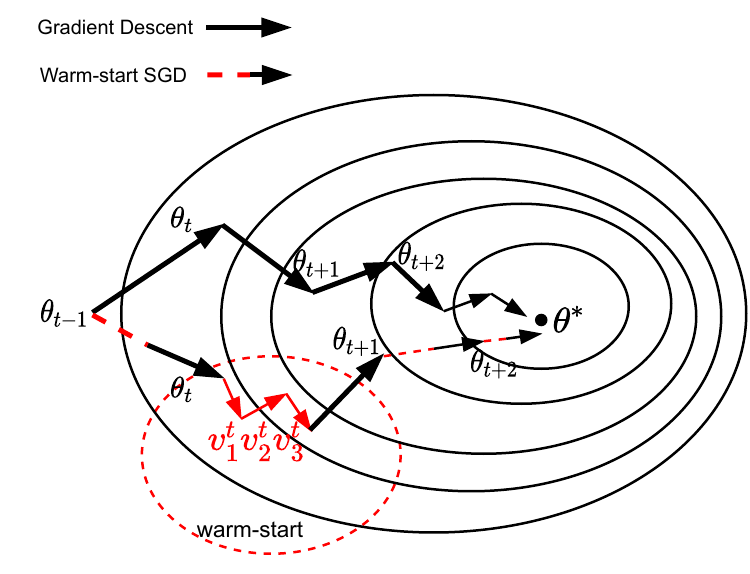}
\caption{A demonstration of gradient trajectory of gradient descent and warm-start QNG. Circle in red indicates the parameter re-initialization for next step parameter update.}
\label{fig:trajectory}
\end{figure}
In this section, we show that QNG corresponding to quantum probability space can be implemented as a classical first-order optimization known as mirror descent.
Then we show the proposed LAWS algorithm and a general WSSGD framework.
\subsection{Quantum Information Geometry of Mirror Descent}
In the seminar work~\cite{Nemirovskii1983}, the Euclidean distance term $\frac{1}{2\eta}\Vert \theta - \theta_t \Vert_2^2$ in Equation~\ref{eq:basic} has been replaced with a general distance function $D_{\Phi}(\cdot, \cdot)$, i.e.,
\begin{equation}\label{eq:bregman}
    D_{\Phi} (\theta_1, \theta_2) = \Phi(\theta_1) - \Phi(\theta_2) - \langle \nabla \Phi(\theta_2), \theta_1 - \theta_2 \rangle
\end{equation}
where $\Phi (\cdot)$ is a carefully chosen continuously differentiable, strictly convex proximity function defined on some convex set.
Notice, $D_{\Phi}(\theta_1, \theta_2) \geq 0$ with $D_{\Phi}(\theta_1, \theta_1) = 0$.
$D_{\Phi}(\cdot, \cdot)$ defined above is also known as \textit{Bregman divergence}, which is  widely used in statistical inference, optimization, machine learning, and information geometry.
As a result, a generalization of stochastic iterative optimization~\ref{eq:basic} has following
\begin{equation}\label{eq:bregmanoptim}
    \theta_{t+1} = \argmin_{\theta} \big \{ \langle \theta, g_t \rangle + \frac{1}{\eta} D_{\Phi}(\theta, \theta_t) \big \}
\end{equation}
The above optimization is known as mirror descent (MD)~\cite{nemirovski2009robust} with proximity function $D_{\Phi}$.
Note, if $\Phi(\theta) = \frac{1}{2}\Vert \theta \Vert_2^2$ convex, then $D_{\Phi}(\theta, \theta_t)= \frac{1}{2}\Vert \theta - \theta_t \Vert_2^2$ yields the standard gradient descent update~\ref{eq:basicsgd}.
In addition, many modern machine learning optimizations such as Vanilla SGD, AdaGrad and Adam~\cite{kingma2014adam} fall into MD~\ref{eq:bregmanoptim} point view.
For example, given Mahalanobis distance $\Phi(\theta) = \theta^{\top} A \theta$ where $A \succ 0$ is a positive (semi)definite matrix, i.e., $A = \sqrt{\sum_{i=1}^{t} g_i^2}$ a sum of all gradients for $t=1$ to $t$.
We have AdaGrad
\begin{equation}
\small
    \theta_{t+1} = \argmin_{\theta} \big \{ \langle \theta, g_t \rangle + \frac{1}{2\eta} (\theta - \theta_t)^{\top} A (\theta - \theta_t)\big \}
\end{equation}
which is equivalent to
\begin{equation}
    \theta_{t+1} = \theta_t - \frac{\eta}{\sqrt{\sum_{i=1}^{t} g_i^2} + \epsilon} \odot g_t
\end{equation}
where $\epsilon$ is a small number, typically set as $10^{-8}$, $\odot$ indicates the element-wise product.
Moreover, if
\begin{equation}
\small
    A = \sqrt{(1 - \beta_2)\sum_{i=1}^{t} \beta_2^{t-i}g_i^2}
\end{equation}
and set $m_t = \beta_1 m_{t-1} + (1 -\beta_1) g_t$ as exponential moving average (EMA) of stochastic gradient $g_t$ with $\beta_1, \beta_2 \in \mathbb R$ (typical values are $\beta_1 = 0.9$ $\beta_2 = 0.999$).
We recover the Adam optimizer
\begin{equation}
\small
    \theta_{t+1} = \theta_t - \frac{\eta}{\sqrt{(1 - \beta_2)\sum_{i=1}^{t} \beta_2^{t-i}g_i^2} + \epsilon}\odot m_t
\end{equation}

In VQA, consider a parametric family of strictly positive probability distributions $p_{\theta}(x)$ parametrized by $\theta \in \mathbb R^d$ where $x\in [N]$ is a set of probability distributions on $N$ elements $[N] = \{1, \cdots, N\}$ and satisfies the normalization condition
\begin{equation}
    \int p_{\theta}(x) dx = 1 \text{ for all } \theta
\end{equation}
Assuming sufficient regularity, the derivatives of such densities satisfy the identity
\begin{equation}\label{eq:firstorderzero}
    \forall t > 0 \quad \int \frac{\partial^t p_{\theta}(x)}{\partial \theta^t} dx = \frac{\partial^t}{\partial \theta^t} \int p_{\theta}(x) dx = \frac{\partial^t 1}{\partial \theta^t} = 0
\end{equation}
To elucidate the geometry of the probability space $P$, we measure the density $p_{\theta}$ changes when
one adds a small quantity $d\theta$ to its parameter. 
It can be achieved in a statistically meaningful way by using the Kullback-Leibler (KL) divergence~\cite{bottou2018optimization}.
Interestingly, KL-divergence is also a instance of Bregman divergence mentioned in Equation~\ref{eq:bregman} by letting proximity function $\Phi(\theta) = \sum_i \theta_i \log (\theta_i)$ result in
\begin{equation}
\small
    D_{\Phi}(\theta, \theta + d\theta) = KL(\theta \Vert \theta+d\theta) = \mathbb E_{p_{\theta}}\Bigg [ \log \bigg ( \frac{p_{\theta}(x)}{p_{\theta + d\theta}(x)}\bigg)\Bigg]
\end{equation}
where $\mathbb E_{p_{\theta}}$ denotes the expectation with respect to the distribution $p_{\theta}$.
Further, we can approximate the divergence with a second-order Taylor expansion such as
\begin{equation}
\small
\begin{split}
    & KL(\theta \Vert \theta+d\theta) = \mathbb E_{p_{\theta}}\big[ \log(p_{\theta}(x)) - \log(p_{\theta + d\theta}(x)) \big]\\
    & \approx -d\theta^{\top} \mathbb E_{p_{\theta}}\Big[ \frac{\partial \log(p_{\theta}(x))}{\partial \theta} \Big] + \frac{1}{2}d\theta^{\top} \mathbb E_{p_{\theta}}\Big[ \frac{\partial^2 \log(p_{\theta}(x))}{\partial \theta^2}\Big] d\theta\\
\end{split}
\end{equation}
Applying the fact that first-order term is 0 shown in Equation~\ref{eq:firstorderzero}, we have
\begin{equation}
    D_{\Phi}(\theta, \theta + d\theta) = KL(\theta \Vert \theta+d\theta) \approx \frac{1}{2}d\theta^{\top} F(\theta) d\theta
\end{equation}
$F(\theta)$ is defined by the Fisher information matrix (FIM)
\begin{equation}
\begin{split}
        F(\theta) &= \mathbb E_{p_{\theta}}\Big[ \frac{\partial^2 \log(p_{\theta}(x))}{\partial \theta^2}\Big]\\
        &= \mathbb E_{p_{\theta}}\Bigg[ \bigg(\frac{\partial\log(p_{\theta}(x))}{\partial \theta}\bigg)\bigg(\frac{\partial\log(p_{\theta}(x))}{\partial \theta}\bigg)\Bigg]
\end{split}
\end{equation}
We notice the second equality of $F(\theta)$ is often preferred because it makes clear that the $F(\theta)$ is symmetric and always positive semidefinite, though not necessarily positive definite.
Finally, we plug the Bregman divergence defined on information entropy $\Phi(\theta) = \sum_i \theta_i \log (\theta_i)$ in MD optimization~\ref{eq:bregmanoptim}, and we have
\begin{equation}\label{eq:mdfisher}
    \theta_{t+1} = \argmin_{\theta} \big \{ \langle \theta, g_t \rangle + \frac{1}{2\eta} (\theta - \theta_t)^{\top} F (\theta - \theta_t)\big \}
\end{equation}
The iterative solution of the above optimization problem~\ref{eq:mdfisher} is
\begin{equation}
    \theta_{t+1} = \theta_t - \eta F^{-1}g_t
\end{equation}
where $F^{-1}$ is the pseudo-inverse of the Fisher information matrix, which recovers the natural gradient descent in~\ref{eq:ngd}.
In the over-parameterized classical deep learning model, $F$ is singular.
To make it invertible, one often adds a non-negative damping term $\delta$ such that $ \theta_{t+1} = \theta_t - \eta (F + \delta I)^{-1}g_t$.

\subsection{Look Around Warm-start Natural Gradient}
The presence of BP becomes one of the major bottlenecks in optimizing VQA, such as deep QNN.
Notably, this does not preclude VQA, allowing for efficient gradient-based optimization.
In section~\ref{subsec:bp}, we discuss two mainstream techniques to mitigate BP.
This work focuses on the optimization solution combined with the QNN parameter initialization strategy.
\subsubsection{Motivation}
Many studies in both classical and quantum regimes have shown that the parameter initialization could significantly improve the model performance and accelerate the training.
Intuitively, a good parameter initialization (i.e., the distribution of initialized parameter close to optimal) requires a large amount of empirical studies, hyper-parameter tuning, and possibly human intervention, which is unproductive.
Therefore, a natural question is raised: \textit{can we perform efficient and effective parameter initialization for QNNs?}
This is the primary motivation behind our approach.
Second, the one-shot model initialization strategy initializes the model only at the beginning of the training process.
However, as the training process proceeds,  BP appears again when we use the gradient-based method to train the model.
So, the questions, such as \textit{can we design a hybrid method that effectively trains QNNs while adopting the superiority of parameter initialization}, or \textit{can we periodically perform parameter initialization during the optimizing phase to mitigate the BP?} are also motivating us to explore various optimization strategies for QNNs.
Besides finding a suitable initialization strategy, we also consider the algorithm's efficiency since computing quantum Fisher information in QNG is expensive, as discussed in section~\ref{subsec:firstorder}.

\subsubsection{Proposed Method}
Here, we present our proposed algorithm shown in Algorithm~\ref{alg:wsngd}.
The key step in the proposed algorithm, in short, is that we perform the initialization after every parameter update instead of only initializing the PQC one at a time.  
The intuition behind this algorithm is that we try to warm-start the natural gradient descent for each iteration. 
Every time the optimizer finds a sub-optimal solution, say $\theta_t$, we utilize this $\theta_t$ and re-initialize the model around $\theta_t$ within a small region.
In this way, we couple the parameter re-initialization and gradient descent method together, and the LAWS finds the (maybe) optimal solution in a ``\textit{look around}" manner.
Later, we generalized the LAWS to accommodate all existing gradient-based methods in Algorithm~\ref{alg:general}.

There are two major advantages when using LAWS.
First,  LAWS could mitigate the BP issue by repeatedly performing the parameter re-initialization, where our empirical results also support this observation.
Second, LAWS adopts a fast convergency speed, and it is more computationally efficient than the QNG.
In the basic LAWS design, the parameter re-initialization is performed by another low-cost first-order optimization in fewer steps (usually $K = 5$).
The expensive quantum Fisher information matrix evaluation is completed after parameter re-initialization.
Assume we require total $T$ iterations for model training. We only evaluate qFIM for $T/K$ rounds. 
Third, we empirically find that LAWS achieves better generalization ability in the classification learning task.
\begin{algorithm2e}[t]
\caption{Look Around and Warm-Start Natural Gradient Descent}
\label{alg:wsngd}
\setstretch{1.2}
\SetAlgoLined
\SetKwInOut{KwIn}{Clients}
\SetKwInOut{KwOut}{Server}
           Input: Objective function $\mathcal C(\theta)$, learning data $\mathcal D$\\
Initialization: $\theta_0$, learning rate $\eta_0$, warm-start learning rate $\mu_0$, warm-start iteration $K$ ($K = 5, 3, \text{or } 2$)\\
\For{$t=1, \cdots, T$}{
$v_{0}^t = \theta_{t-1}$\\
\For{$k = 1, \cdots, K$}{
Draw sample from batch data $\xi \sim \mathcal D_b$\\
$v_{k}^t = v_{k-1}^t - \mu_k \nabla C(\theta; \xi)$\\
}
$\theta_{\text{warm-start}}^t = v_k^t$\\
Compute natural gradient $F_t = \text{FisherIM}(\theta_{\text{warm-start}}^t)$\\
Compute new gradient $g_t = \theta_{\text{warm-start}}^t - \theta_{t-1}$\\
Update parameter:\\
$\theta_{t} = \theta_{\text{warm-start}}^t - \frac{\eta_t}{K} F_t^{-1} g_t$\\
}
\end{algorithm2e}
\subsubsection{Implementation Details}\label{subsec:implementation}
The implementation of LAWS is simple and is compatible with all existing gradient-based optimization frameworks.
Therefore, how to effectively and efficiently perform warm-start (parameter re-initialization) is the key challenge in LAWS's design.
Here, effective warm-start means we expect the periodically re-initialized parameter to be close to the region where the possible optimal parameter resides.
Furthermore, efficient warm-start means parameter re-initialization should be computational inexpensive and possibly without hyper-parameter tuning.
To this end, the design of warm-start is based on a stochastic procedure, where the re-initialized parameter is sampled from a set of stochastic gradients.
Fig~\ref{fig:trajectory} demonstrates the optimization trajectory of LAWS compared to the original QNG.
We search gradients for fewer steps around the current optimal $\theta_t$ and then perform a natural gradient descent step ($F_t^{-1}$) on the accumulation of the previous gradient ($v_{k}^t - \theta_t$) at the re-initialized parameter point $v_{k}^{t}$.

We present two different warm-start strategies.
The first one uses a K-step (K usually small, such as $K=5$) inner loop (as the Algorithm~\ref{alg:wsngd} shows) to compute a set of K consecutive gradients such as $\mathcal G_K^t = \{g_{1}^t, g_{2}^t,\cdots,g_{K}^t\}$.
Then, we compute a weighted average of gradients in $\mathcal G_K^t$ as a warm-start point of the next iteration
\begin{equation}
    \theta_{\text{warm-start}} = \theta_{t-1} + \frac{1}{K}\sum_{k=1}^{K} g_{k}
\end{equation}
for all $g_k^t \in \mathcal G_K^t$ such that
\begin{equation}\label{eq:gradientset}
     \mathcal G_K^t = \big \{ \nabla \mathcal C(v_{k}^t, \xi) | \xi \sim \mathcal D_b \big\}
\end{equation}
where each $v_{k}^t$ is computed as line 9 in Algorithm~\ref{alg:wsngd}.
The second one also uses a K-step inner loop to sample gradient candidates.
But one significant difference compared to the first method is that sample gradient candidates are computed with respect to the same model parameter at current step $t-1$, say $\theta_{t-1}$.
Mathematically, we have
\begin{equation}
    \theta_{\text{warm-start}} = \theta_{t-1} + \frac{1}{K}\sum_{k=1}^{K} \nabla \mathcal C(\theta_{t-1}, \xi) \text{ where } \xi \sim \mathcal D_b
\end{equation}
The two methods above ensure that the warm-start parameter is not far from the current optimal solution.
The re-initialization is achieved using scholastic gradients sampled from randomly selected samples in (min) batch training data.
It is worth mentioning that the FIM ($F$) in LAWS is evaluated on the warm-start parameter $\theta_{\text{warm-start}}$.
In Algorithm~\ref{alg:wsngd}, line 10 (or Algorithm~\ref{alg:general}, line 8), we explicitly show the $\theta_{\text{warm-start}} = v_{k}^t$ for clear presentation.
In an actual implementation, we directly use $v_{k}^t$ as a warm-start parameter to save the memory usage.

Another discussion for Algorithm~\ref{alg:wsngd} is the gradient defined as $g_t = \theta_{\text{warm-start}}^t - \theta_{t-1}$ (line 11).
LAWS simply uses the accumulation of all past gradients in $\mathcal G_K^t$ defined in~\ref{eq:gradientset}.
The past gradients' accumulation re-scales the gradient and may increase the magnitude of the new one compared to using a single gradient.
We believe this is also a key point in making LAWS mitigate the BP.
One could also perform this accumulation by using exponential moving average (EMA), similar to Adam settings, to apply more smooth gradient such as
\begin{equation}
    g_t = (1 - \beta)\sum_{i=1}^{K} \beta^{K-i} g_{i}^{t} \text{ for all } g_{i}^t \in \mathcal G_{K}^t
\end{equation}
where we introduce a new hyper-parameter $\beta \in (0,1)$.
Obviously, the process can be replaced with some gradient noise reduction techniques widely used in various classical machine learning applications.
This is beyond the scope of this literature, and we leave this for future work.
We empirically evaluate the above-mentioned warm-start strategies.
More detailed results and analysis are shown in section~\ref{sec:exp}.

We notice that the proposed LAWS belongs to a certain first-order optimization in modern classical learning regime so-called \textit{Lookahead optimizer}.
Based on their extraordinary work, we propose a general warm-start framework for VQA in the next section.
\begin{algorithm2e}[t]
\caption{Generalized Warm-start Stochastic Gradient Descent (WSSGD)}
\label{alg:general}
\setstretch{1.2}
\SetAlgoLined
Input: Objective function $\mathcal C(\theta)$, learning data $\mathcal D$, warm-start optimizer $\mathcal W$, reparameterization coefficient function $\Delta$ \\
Initialization: $\theta_0$, warm-start learning rate $\mu_0$, warm-start iteration $K$ \\
\For{$t=1, \cdots, T$}{
$v_{0}^t = \theta_{t-1}$\\
\For{$k = 1, \cdots, K$}{
$v_{k}^{t} = \mathcal W \big( \mathcal C(\theta), v_0^{t}, \mu_0, \mathcal D_m \big)$\\
}
$\theta_{\text{warm-start}}^t = v_k^t$\\
Update parameter:\\
$\theta_{t+1} = \Delta_t \theta_{t} +  (1 - \Delta_t) \theta_{\text{warm-start}}^t$\\
}
\end{algorithm2e}

\subsubsection{Generalized Warm-start Algorithm}
In the classical machine learning study,~\cite{zhang2019lookahead} proposed a new optimization algorithm named Lookahead.
Lookahead is orthogonal to those aforementioned approaches~\cite{kingma2014adam} due to the different parameter update settings.
The core idea of Lookahead is to maintain two kinds of model parameters, i.e., ``fast parameter" $v_k^t$ and ``slow parameter" $\theta_t$, and jointly update them.
Specifically, the inner loop takes the slow weights ($\theta_{t-1}$) as initial point and updates the fast weights ($v_{k}^t$) $K$ times to receive $v_{K}^t$; while the outer loop updates the slow weights as 
\begin{equation}\label{eq:lookahead}
    \theta_t = (1 - \alpha)\theta_{t-1} + \alpha v_{K}^t \text{ } \alpha \in (0,1)
\end{equation}
Any standard optimizer, e.g., Vallina SGD, AdaGrad, and Adam, can serve as the inner-loop optimizer.
In our speech, the inner-loop act as a warm-start initialization.
In this way, the Lookahead optimizer achieves remarkable performance improvement over the standard optimizer.
Further, due to its simplicity in implementation, negligible computation and memory cost, and compatibility with almost current ML libraries, Lookahead has been widely adopted.

Interestingly, we find LAWS also falls into this line of research.
In algorithm~\ref{alg:wsngd}, let $\lambda_t = \eta_t / K$, we compute the $\theta_t$ as
\begin{equation}
    \begin{split}
        \theta_t &= \theta_{\text{warm-start}}^t- \lambda_t F_t^{-1} g_t\\
        & = \theta_{\text{warm-start}}^t - \lambda_t F_t^{-1} (v_{k}^t - \theta_{t-1})\\
        & = \lambda_t F_{t}^{-1} \theta_{t-1} + (1 - \lambda_t F_t^{-1}) \theta_{\text{warm-start}}^t\\
    \end{split}
\end{equation}
the last equality is due to $v_k^t = \theta_{\text{warm-start}}^t$.
From the above derivation, We see that the mathematical difference between LAWS and Lookahead is: we replace $\alpha$ in Lookahead~\ref{eq:lookahead} to some value such as $$\alpha = 1 - \lambda_t F_t^{-1}$$, where is not a fixed real coefficient but a Fisher information related quantity.

We can reveal a couple of insights from this observation.
First, this modification makes the proposed algorithm LAWS significantly different from the Lookahead in terms of parameter update and design logic.
Second, when the Fisher information matrix degenerates to the identity matrix $I$, we recover the standard Lookahead optimizer.
On the other hand, one could use the Hessian matrix to replace $F$, where we have a Newton-type optimizer.
Thrid, essentially, the information on the curvature of the cost surface is encoded in FIM for the Riemannian manifold.
FIM can be interpreted as some specific "step size" in first-order optimization. 
In LAWS, the FIM reparameterizes the model according to its second-order information.
The coefficient $\alpha$, therefore, may not be a fixed real number now.

To this end, we propose our unified framework WSSGD for a warm-start stochastic gradient descent algorithm for QNNs as shown in Algorithm~\ref{alg:general}.
We first employ a generalized optimizer $\mathcal W$ for the warm-start inner loop.
The choice of such an optimizer heavily affects re-initialization and model performance.
As reported in~\cite{zhang2019lookahead}, WSSGD may benefit from a larger learning rate in the inner loop.
In other words, we could use a larger step size $\mu$.
We also propose the general form for reparameterization coefficient $\Delta_t$ as a function of gradient, for example
\begin{equation}
\begin{split}
    \text{Lookahead: }&\Delta_t = 1 -\alpha  \\
    \text{WS-SGD: }&\Delta_t = 1 - \lambda_t F_t \\
    \text{Adam-like SGD: } &\Delta_t = 1 - \lambda_t \sqrt{\sum_k{\nabla C(v_k^t, \xi)^2}}
\end{split}
\end{equation}
Moreover, WSSGD is simple and compatible with many VQA libraries such as Pennylane, Tenoerflow Quantum, etc.
Our empirical results are present in section~\ref{sec:exp}, we conclude that WSSGD achieves faster convergence rates
(i.e., smaller optimization error), and enjoys smaller generalization error.
We release our source code implemented by Pennylane at \url{https://github.com/taozeyi1990/LAWS}

\subsection{Convergence Analysis}
In this section, we present the convergence and generalization analysis of the proposed algorithm.
We first provide some useful definitions and assumptions which have been widely adopted in classical machine learning.
We provide the analysis of the convergence for both convex and non-convex objective functions $\mathcal C(\theta)$.
We start by showing the proof of convergence on the convex problem to give some intuition first and then give the proof on a more realistic non-convex problem.
\subsubsection{Assumption}
The assumptions we are making are
\begin{assumption}\label{as1}
(Bounded gradient). The function $\mathcal C(\theta)$ has bounded (stochastic) gradients, i.e., for any $\theta \in R^d$ we have
\begin{equation}
    ||\nabla \mathcal C(\theta, \xi) ||_2 \leq G \text{ for all } \xi \sim \mathcal D
\end{equation}
\end{assumption}
\begin{assumption}\label{as2}
(L-Lipschitz smooth). The function $\mathcal C(\theta)$ is L-Lipschitz smooth i.e.
\begin{equation}
    ||\nabla \mathcal C(\theta) - \nabla \mathcal C(\mu) ||_2 \leq L||\theta - \mu||\text{ for all } \theta, \mu\in R^d
\end{equation}
\end{assumption}
\begin{assumption}\label{as3}
(M-Lipschitz continuous). The function $\mathcal C(\theta)$ is L-Lipschitz continous i.e.
\begin{equation}
    ||\mathcal C(\theta) - \mathcal C(\mu) ||_2 \leq M||\theta - \mu||\text{ for all } \theta, \mu\in R^d
\end{equation}
\end{assumption}
We also define Polyak-Lojasiewicz (PL) condition as
\begin{definition}\label{df1}
(PL Condition) Let $\theta^* \in \argmin_{\theta} \mathcal C(\theta)$. We say a function $\mathcal C(\theta)$ satisfies $\sigma$-PL condition if
\begin{equation}
    \sigma (\mathcal C(\theta) - \mathcal C(\theta^*)) \leq \Vert \nabla \mathcal C(\theta)||^2
\end{equation}
with some constant $\sigma$.
\end{definition}
The above assumptions and definition can be easily obtained and verified in VQE.
Now, we show our theoretical results below in the Theorem~\ref{thm1} and Theorem~\ref{thm2}.
% Think about a cost function with two parameter $\boldsymbol{\theta} =\{\theta_1, \theta_2\}$ system defined by square of the trace distance $D_{\Tr}$ between $|0\rangle$ and $|\psi\rangle = U^{\dagger}(\theta)|0\rangle$.
% We have cost function
% \begin{equation*}
%     \mathcal C (\boldsymbol{\theta}) = \Tr [O U(\theta)|0\rangle\langle 0| U^{\dagger}(\theta)]
% \end{equation*}
% where $O = I - |0\rangle\langle 0|$.
% The close form of above cost is
% \begin{equation}
%     C(\boldsymbol{\theta}) = 1 - \prod_{j=1}^2 \cos^2 \frac{\theta_j}{2}
% \end{equation}
% The global minimum of $\mathcal C (\boldsymbol{\theta})=0$ is at $(0,0)$.
% The gradient of $\partial \mathcal C (\boldsymbol{\theta})$ is $[ \frac{1}{2}\sin{\theta_1}\cos^2\frac{\theta_2}{2}, \frac{1}{2}\sin{\theta_2}\cos^2\frac{\theta_1}{2}]^{\top}$ is bounded.

\subsubsection{Convex Objective function}
Given a dataset $\mathcal D = \{(x_i, y_i)\}_{i=1}^n$ where $(x_i, y_i)$ is is drawn from an unknown distribution, one often minimizes the empirical risk $\mathcal L (\theta) = \frac{1}{n}\sum_{i=1}^n \mathcal C(\theta, x_i, y_i)$ via a randomized algorithm, e.g. SGD, to find an estimated optimum $\theta_T \in \argmin_{\theta} L(\theta)$.
However, this empirical solution $\hat{\theta}$, differs from the desired optimum $\theta^*$ of the population risk
\begin{equation}
    \theta^{*} \in \argmin_{\theta} L(\theta, \mathcal D) = \mathbb E_{x,y\sim \mathcal D} [\mathcal C(\theta, x)]
\end{equation}
To begin with, we first investigate the convergence performance of WSSGD when its warm-start optimizer $\mathcal W$ is SGD.
We summarize our main results in Theorem~\ref{thm1} below.
\begin{theorem}\label{thm1}(Convex) Suppose the objective function $\mathcal C(\theta)$ is $gamma$-strongly convex, M-Lipschitz continuous, and L-Lipschitz smooth w.r.t., $\theta$. Let $\theta^* = \argmin_{\theta}\mathcal C(\theta)$. Let war-start learning rate $\mu_k^t=\frac{c_0}{((t-1)k + K +2)}, c_0 \in (0,1]$, the optimization error of the output $\theta_T$ of WS-SGD satisfies
\begin{equation}
\begin{split}
        \mathbb E[\mathcal C(\theta_T) - \mathcal C\theta^*)] &\leq \frac{e^{2\Delta}L(k+2)^{2\Delta}}{2((T+1)K +2)^{2\Delta}}\Vert \theta_0 - \theta^*\Vert^2\\
        &+ \frac{16LG^2}{c_0^2((T+1)K + 2)^{2(1 - \Delta)}(2\Delta - 1)}
\end{split}
\end{equation}
\end{theorem}
% \begin{figure*}[t]
% \centering
% \subfigure[]{\includegraphics[width=1.\columnwidth]{figs/3bits.pdf}\label{fig:toy:a}}
% \subfigure[]{\includegraphics[width=1.\columnwidth]{figs/3bits_results.pdf}\label{fig:toy:b}}
% \caption{Evaluation on randomly designed PQC}
% \vspace{-0.5cm}
% \label{fig:toy}
% \end{figure*}
% \begin{figure*}[t]
% \centering
% \subfigure[]{\includegraphics[width=1.\columnwidth]{figs/HM.pdf}\label{fig:mh:a}}
% \subfigure[]{\includegraphics[width=1.\columnwidth]{figs/HM_EXP.pdf}\label{fig:mh:b}}
% \caption{The PQC design and evaluation on Hydrogen VQE}
% \label{fig:mh}
% \end{figure*}
\begin{proof}
(Proof sketch) Due to the page limitation, we present the proof sketch to show some core results.
This proof mainly follows the proof of Theorem 1 in~\cite{zhou2021towards}.
We first bound the inner loop $\mathbb E\Vert v_{K}^t - \theta^*\Vert$, where we have $\mathbb E\big[\Vert v_{K}^t - \theta^*\Vert^2\big] \leq (1 - \gamma \mu_{K-1}^t) \mathbb E\big[\Vert v_{K-1}^t - \theta^*\Vert^2\big] + (\mu_{K-1}^t G)^2$.
Let $\mu_{k}^t = \frac{c_0}{((t-1)k + K +2)}$, for some constant $c_0 >0$, we roll the above recurrence relation from $k = 1$ to $K$, we have $\mathbb E\big[\Vert v_{K}^t - \theta^*\Vert^2\big] \leq \bigg ( (\frac{(t-1)k + 2}{tk+2})^2 \mathbb E\big[\Vert \theta_{t-1} - \theta^*\Vert^2\big] + \frac{16c_0G^2}{c_0^2(tk+2)^2}\bigg)$.
The outer loop parameter update follows $\mathbb E\big[\Vert \theta_{t} - \theta^*\Vert^2\big] \leq \Delta\mathbb E\big[\Vert \theta_{t-1} - \theta^*\Vert^2\big] + (1- \Delta)\mathbb E\big[\Vert v_{K}^t - \theta^*\Vert^2\big]$ according to line 10 in Algorithm~\ref{alg:general}.
We have $\mathbb E\big[\Vert \theta_{t} - \theta^*\Vert^2\big] \leq \bigg(\Delta + (1 -\Delta)\big(\frac{(t-1)K + 2}{tK+2}\big)^2\bigg) \mathbb E\big[\Vert \theta_{t-1} - \theta^*\Vert^2\big] + \frac{16c_0(1-\Delta)G^2}{c_0^2(tK+2)^2}$.
Again, unwinding this recurrence relation from $t=1$ to $T$, yields $\mathbb E\big[\Vert \theta_{T} - \theta^*\Vert^2\big] \leq \frac{e^{2\Delta} (k+2)^{2\Delta}}{((T+1)K +2)^{2\Delta}}\Vert \theta_0 - \theta^*\Vert + \frac{32G^2}{c_0^2((T+1)K + 2)^{2(1 - \Delta)}(2\Delta - 1)}$.
Since $\mathcal C(\theta)$ is M-smooth and $\theta^*$ is its optimum, we have $\mathbb E[\mathcal C(\theta_T) - \mathcal C(\theta^*)] \leq \frac{M}{2} \mathbb E[\Vert \theta_T - \theta^*\Vert^2]$, yields the result in Theorem~\ref{thm1}.
\end{proof}
\subsubsection{Non-Convex Objective function}
To prove no-convex objective function, we use the Polyak-Lojasiewicz condition defined in~\ref{df1}.
\begin{theorem}\label{thm2}(Non-Convex) Suppose the objective function $\mathcal C(\theta)$ is M-Lipschitz continuous, and L-Lipschitz smooth w.r.t., $\theta$. In addition, suppose $\mathcal C(\theta)$ satisfies $\sigma$-PL condition. Let $\mu_{k}^t = \frac{1}{tK+t+1}$, we have
\begin{equation}
\begin{split}
        \mathbb E[\mathcal C(\theta_T) - \mathcal C(\theta^*)] &\leq \frac{4}{(TK + 1)^{2\Delta}}\mathbb E \big[ \mathcal C(\theta_0) - \mathcal C(\theta^*)\big]\\
        &+ \frac{2\Delta M G^2C_0}{(TK + 1)^{2\Delta -1}}
\end{split}
\end{equation}
where $C_0 = \Delta + (1- \Delta)(K-1)$.
\end{theorem}
We again omit the complete proof of this theorem.
The proof sketch is we first bound $\mathbb E\big[\Vert\mathcal C (v_K^t) - \mathcal C(\theta^*)\Vert\big]$, then we use the relation of $v_K^t$ and $\theta_t$ defined in the Algorithm~\ref{alg:wsngd} line 10 to derive the final bound of $\mathbb E[\mathcal C(\theta_T) - \mathcal C(\theta^*)]$.
In the next section, we present the numerical simulation results of LAWS, WS-SGD, and their variants.
\section{Numerical Simulations}\label{sec:exp}
\begin{figure*}[t]
    \centering
    \includegraphics[width=1.9\columnwidth]{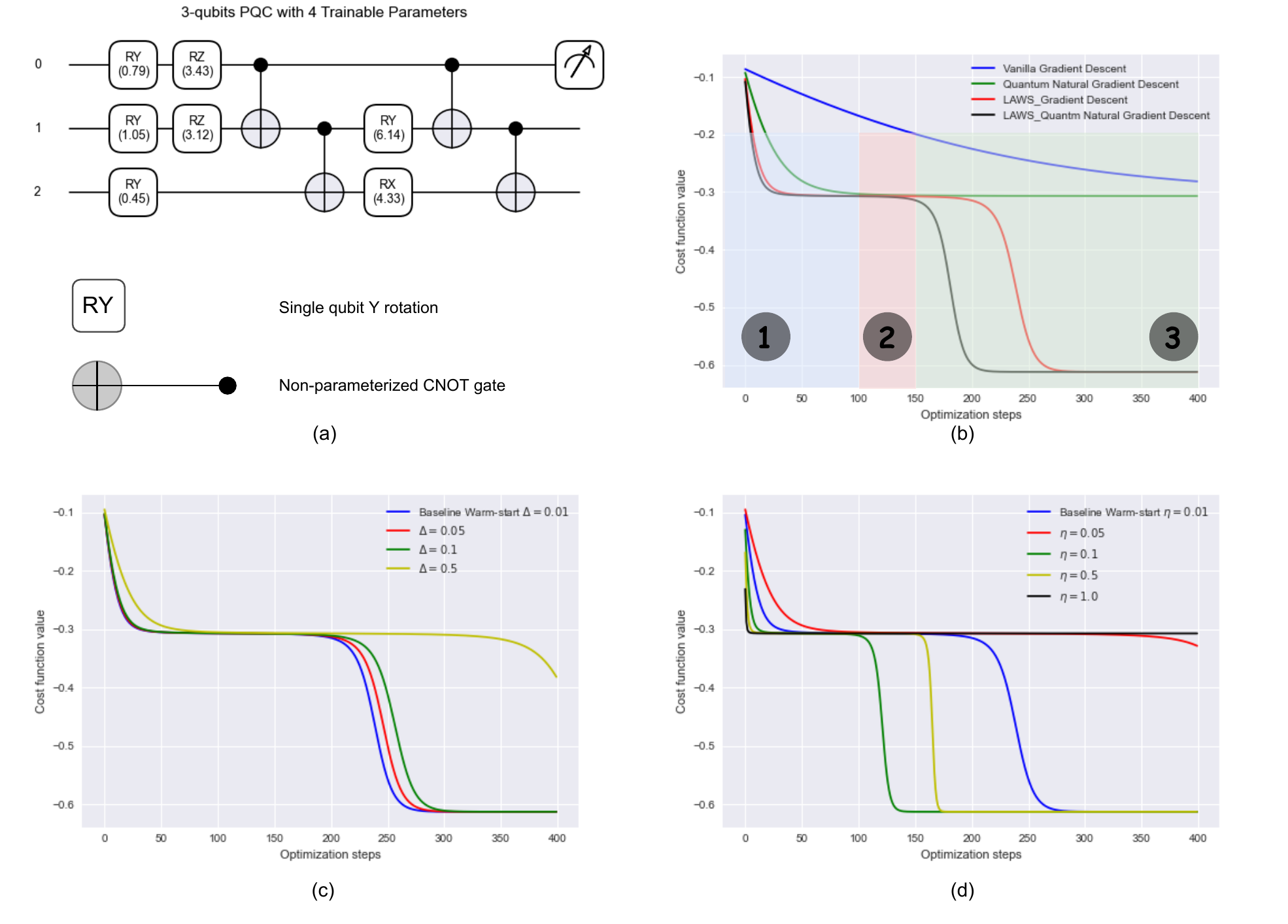}
    \caption{Evaluation on randomly designed PQC}
    \label{fig:randomPQC}
\end{figure*}

\begin{figure*}[t]
    \centering
    \includegraphics[width=2.1\columnwidth]{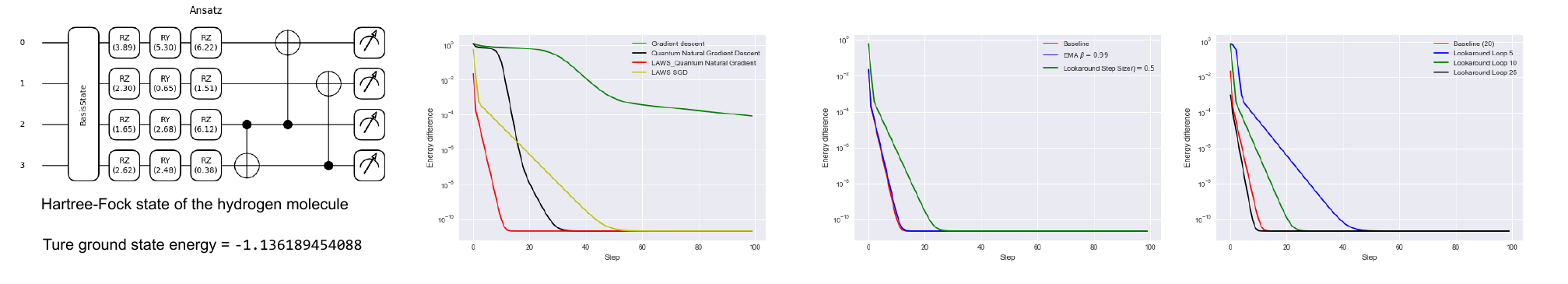}
    \caption{Performance on Hydrogen VQE}
    \label{fig:HM}
\end{figure*}

\begin{figure*}[t]
\centering
\subfigure[Cost Value]{\includegraphics[width=.66\columnwidth]{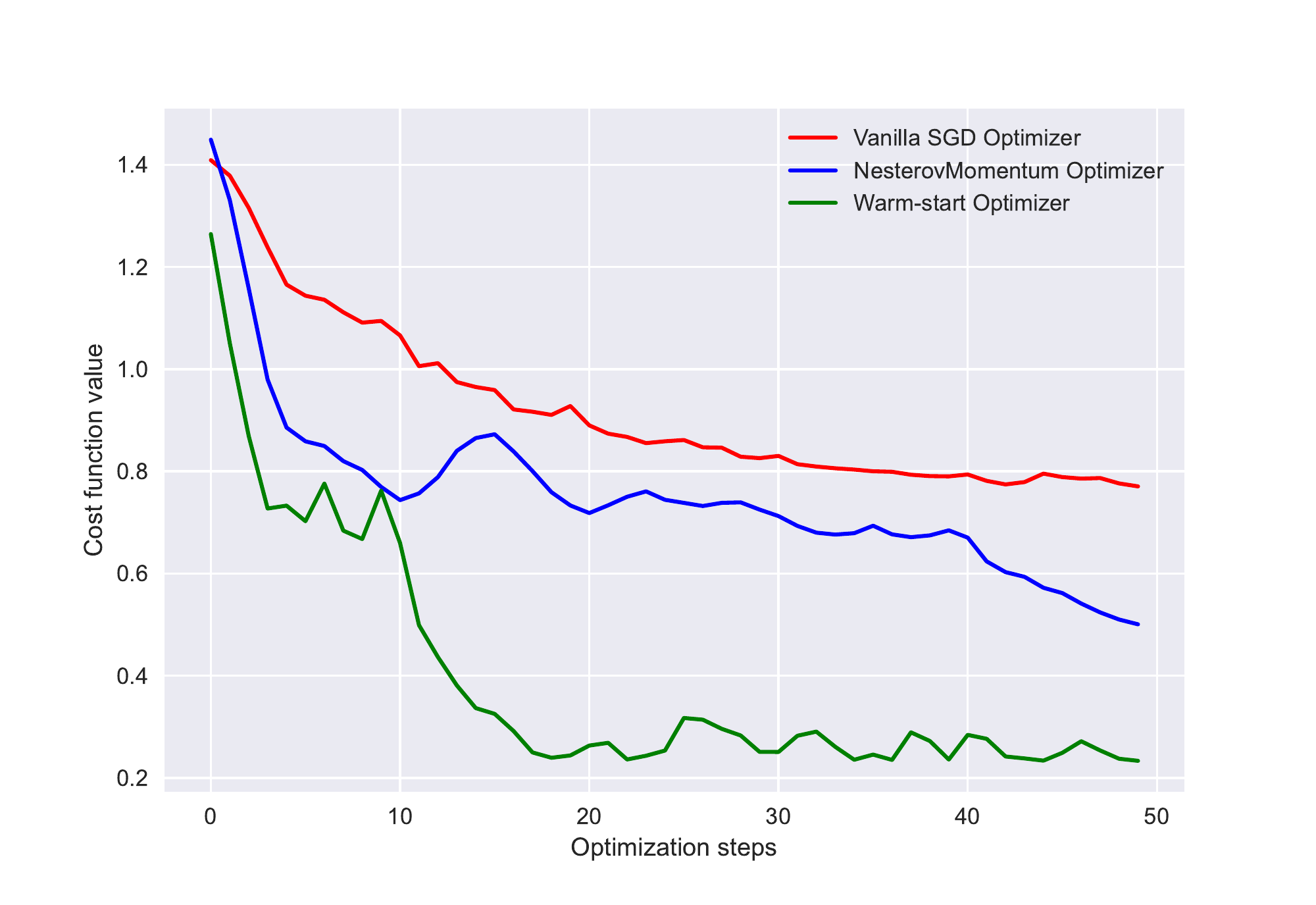}\label{fig:class:a}}
\subfigure[Training Accuracy]{\includegraphics[width=.66\columnwidth]{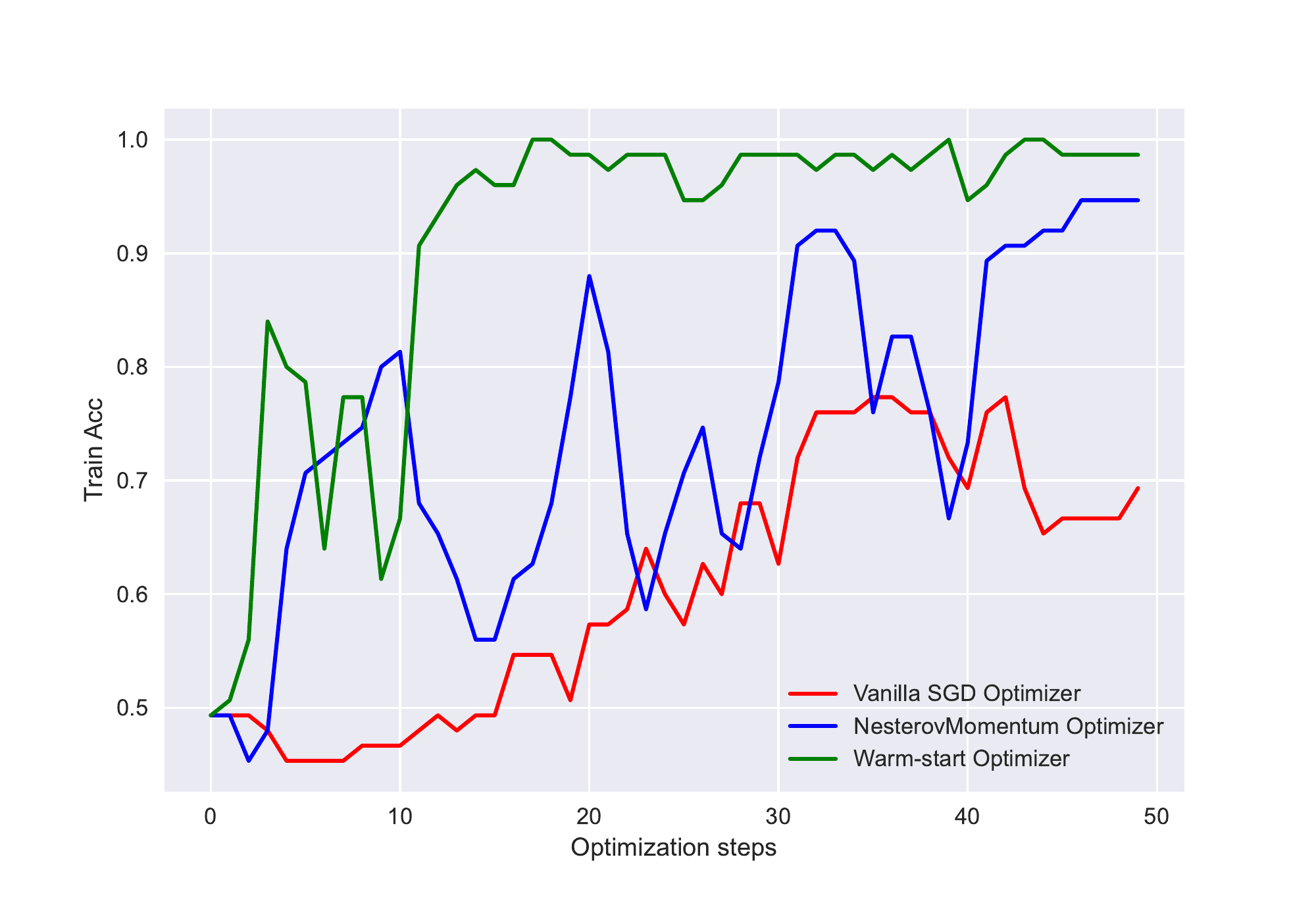}\label{fig:class:b}}
\subfigure[Validation Accuracy]{\includegraphics[width=.66\columnwidth]{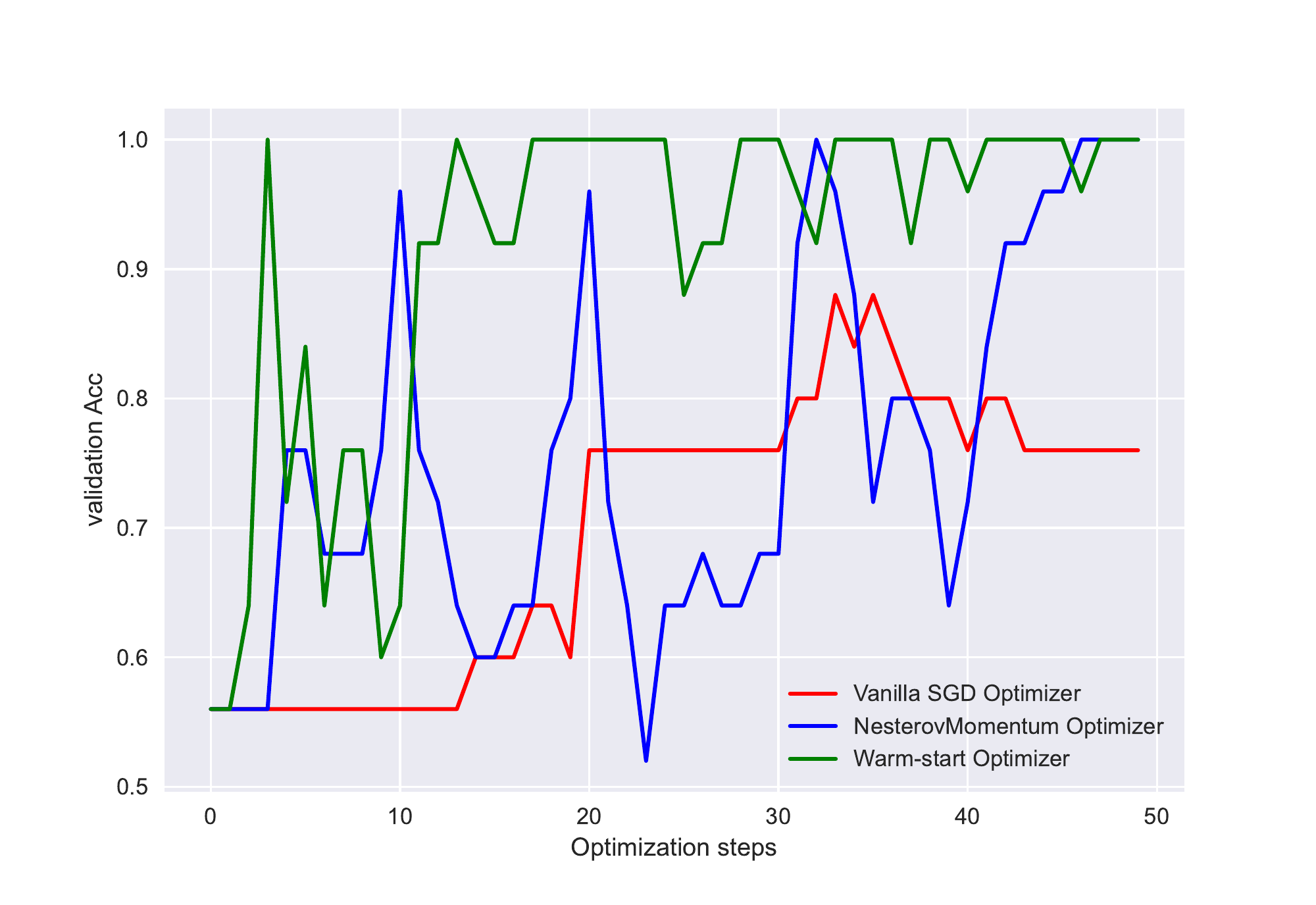}\label{fig:class:c}}
\subfigure[Classifier Trained via SGD]{\includegraphics[width=.66\columnwidth]{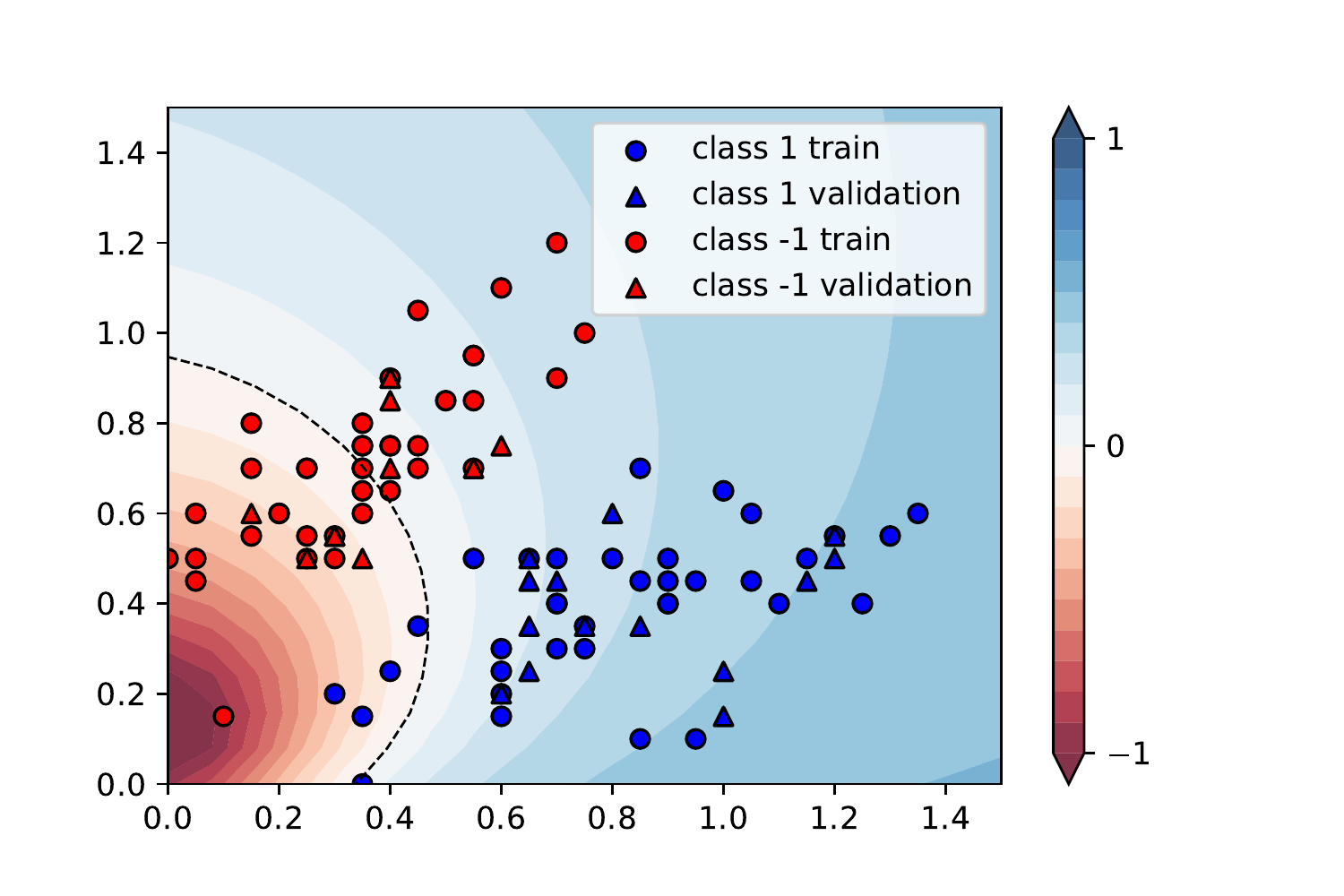}\label{fig:class:d}}
\subfigure[Classifier Trained via Nesterov]{\includegraphics[width=.66\columnwidth]{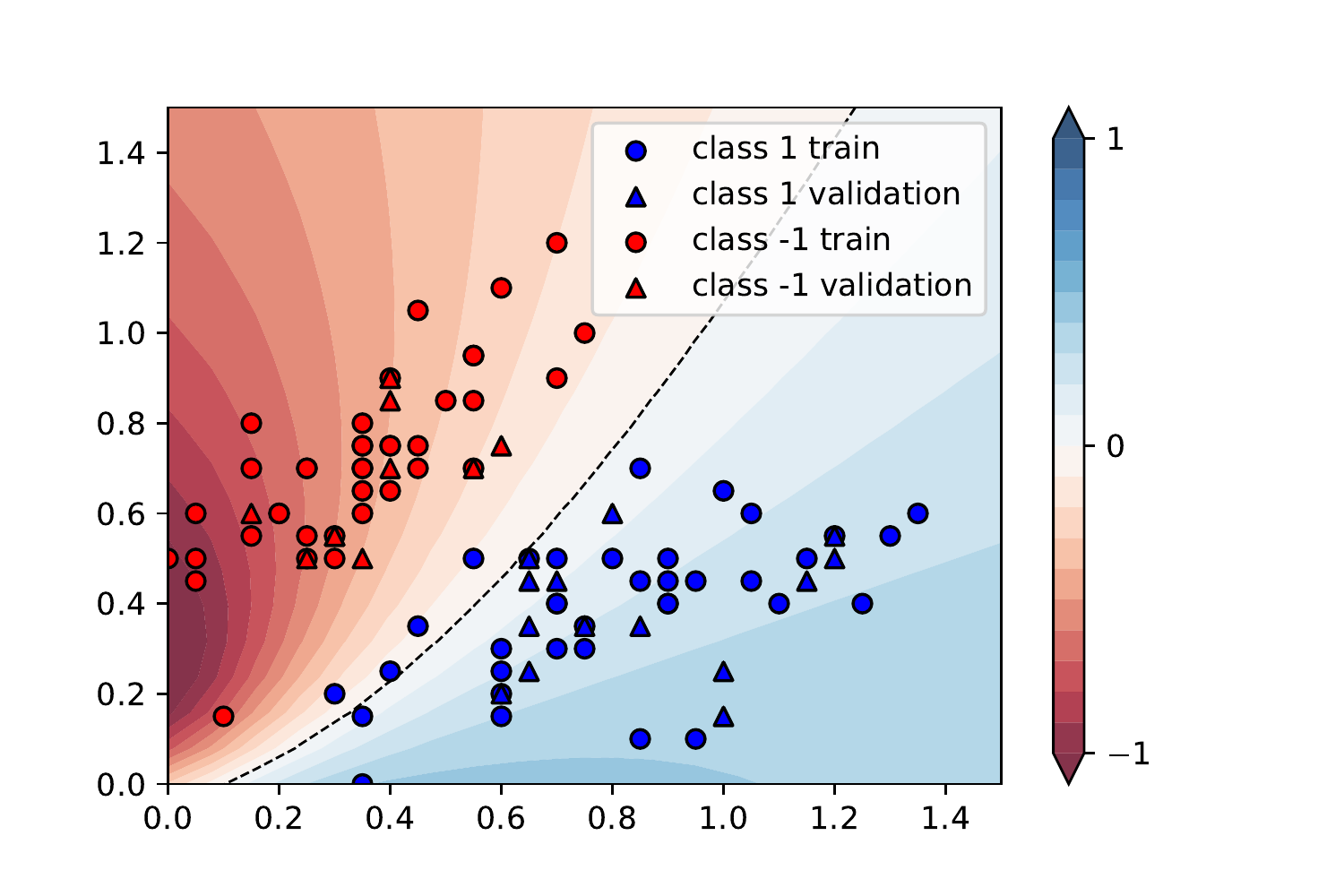}\label{fig:class:e}}
\subfigure[Classifier Trained via LAWS]{\includegraphics[width=.66\columnwidth]{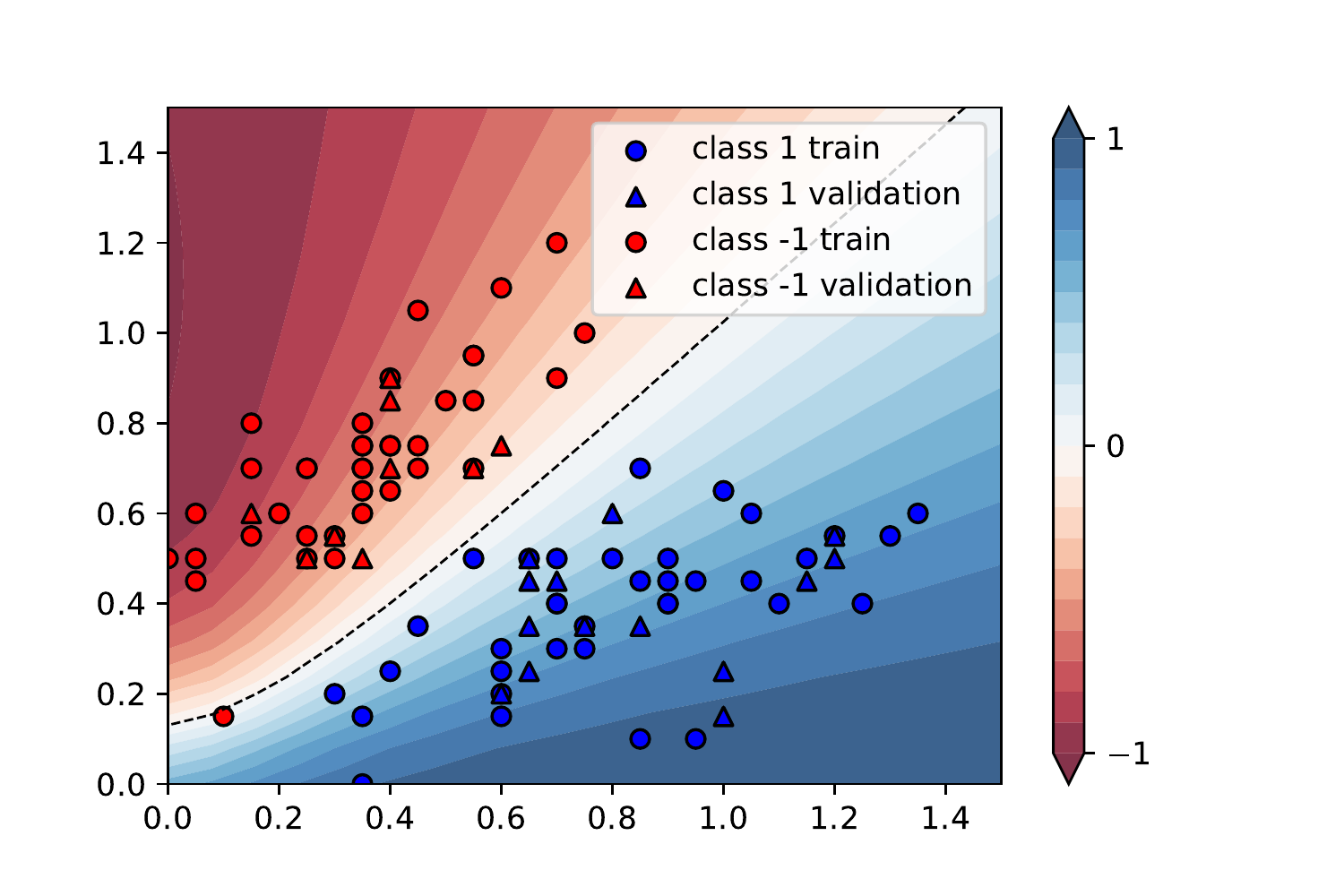}\label{fig:class:f}}
\caption{Variational classifier for Iris classification task: SGD vs. Nesterov vs. LAWS}
\label{fig:class}
\end{figure*}

\begin{figure*}[t]
\centering
\subfigure[costs of WS-SGD's variants]{\includegraphics[width=.5\columnwidth]{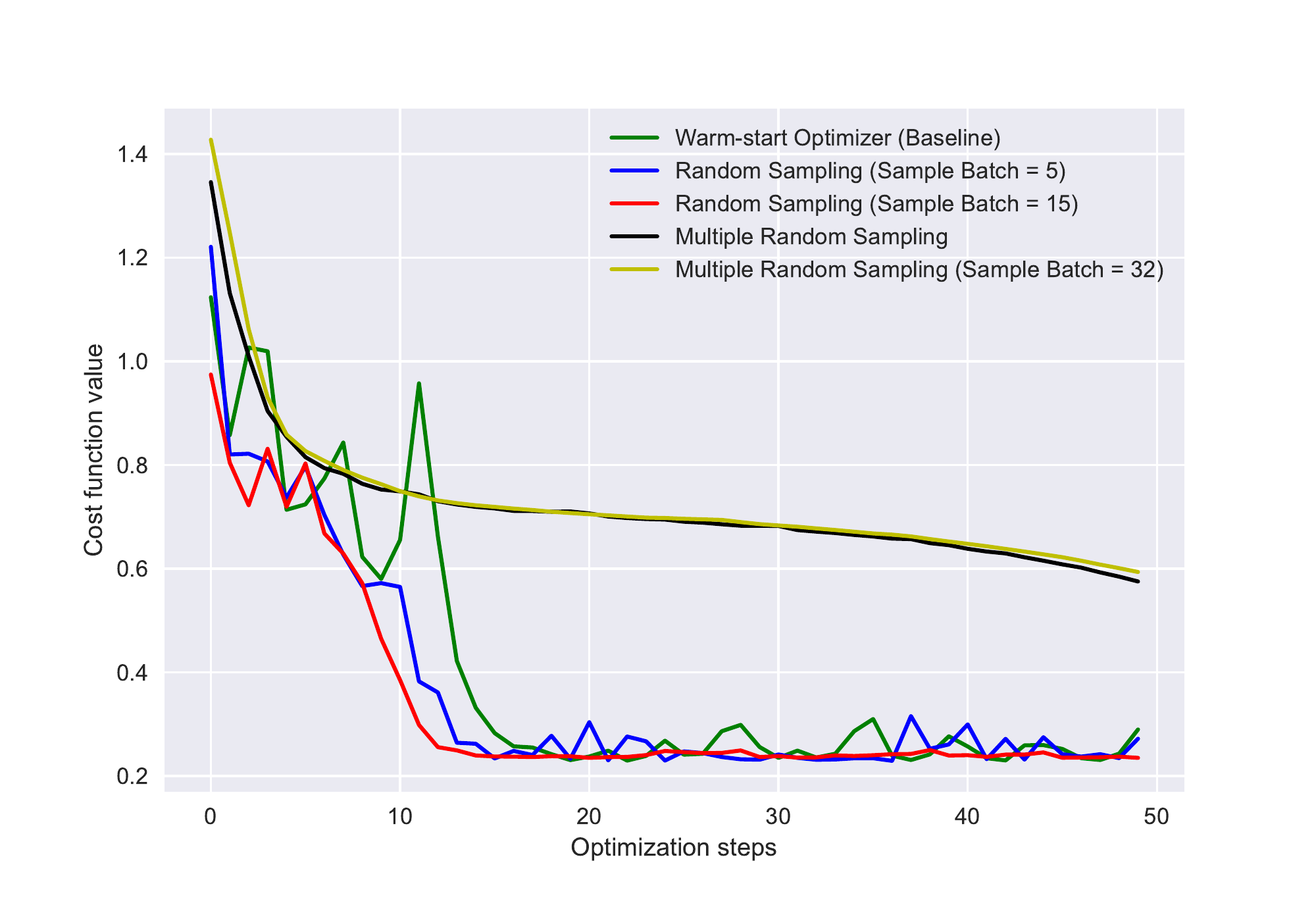}\label{fig:va:a}}
\subfigure[Train Acc of WS-SGD's variants]{\includegraphics[width=.5\columnwidth]{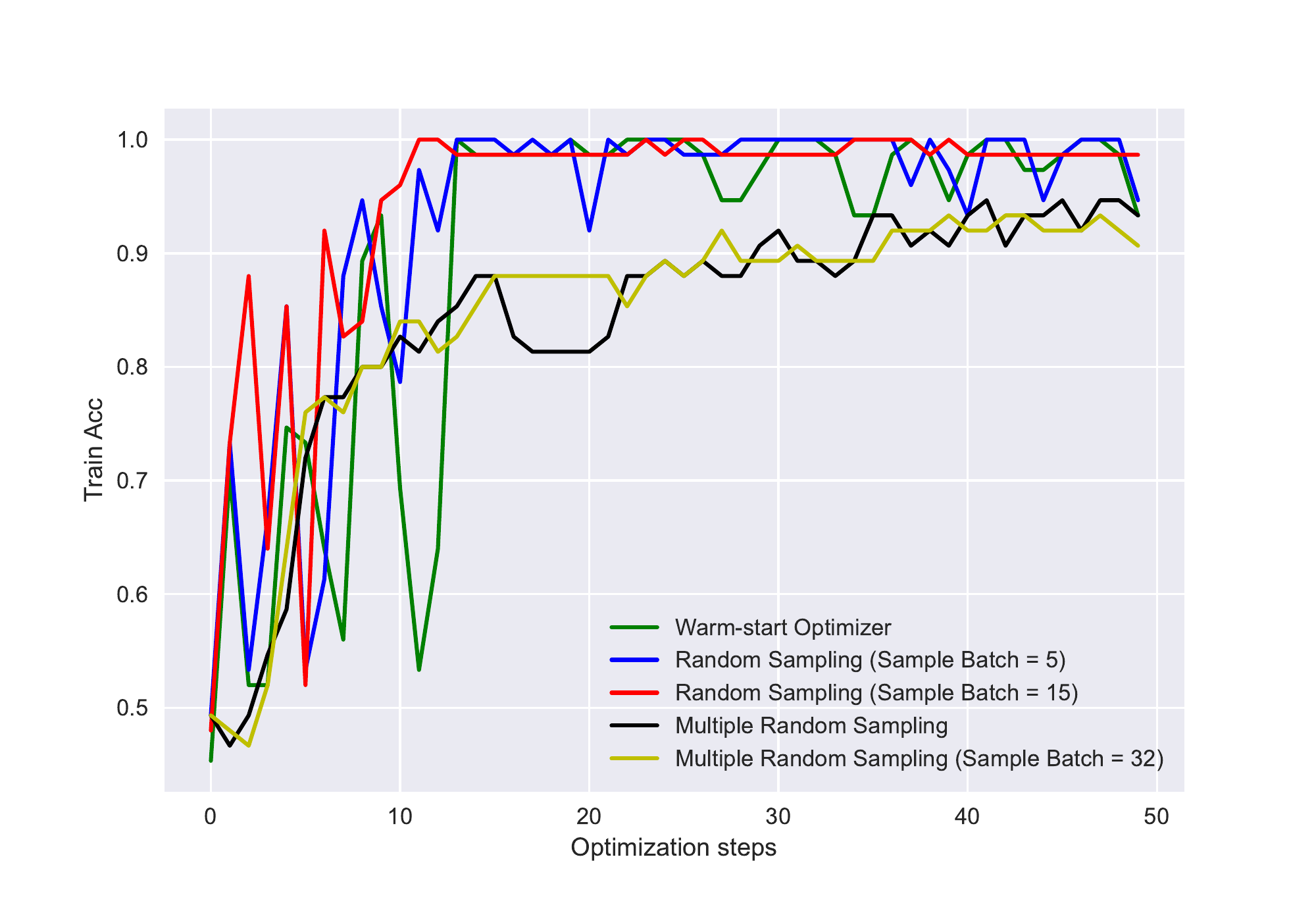}\label{fig:va:b}}
\subfigure[WS-SGD Warm-start with Random Sample]{\includegraphics[width=.5\columnwidth]{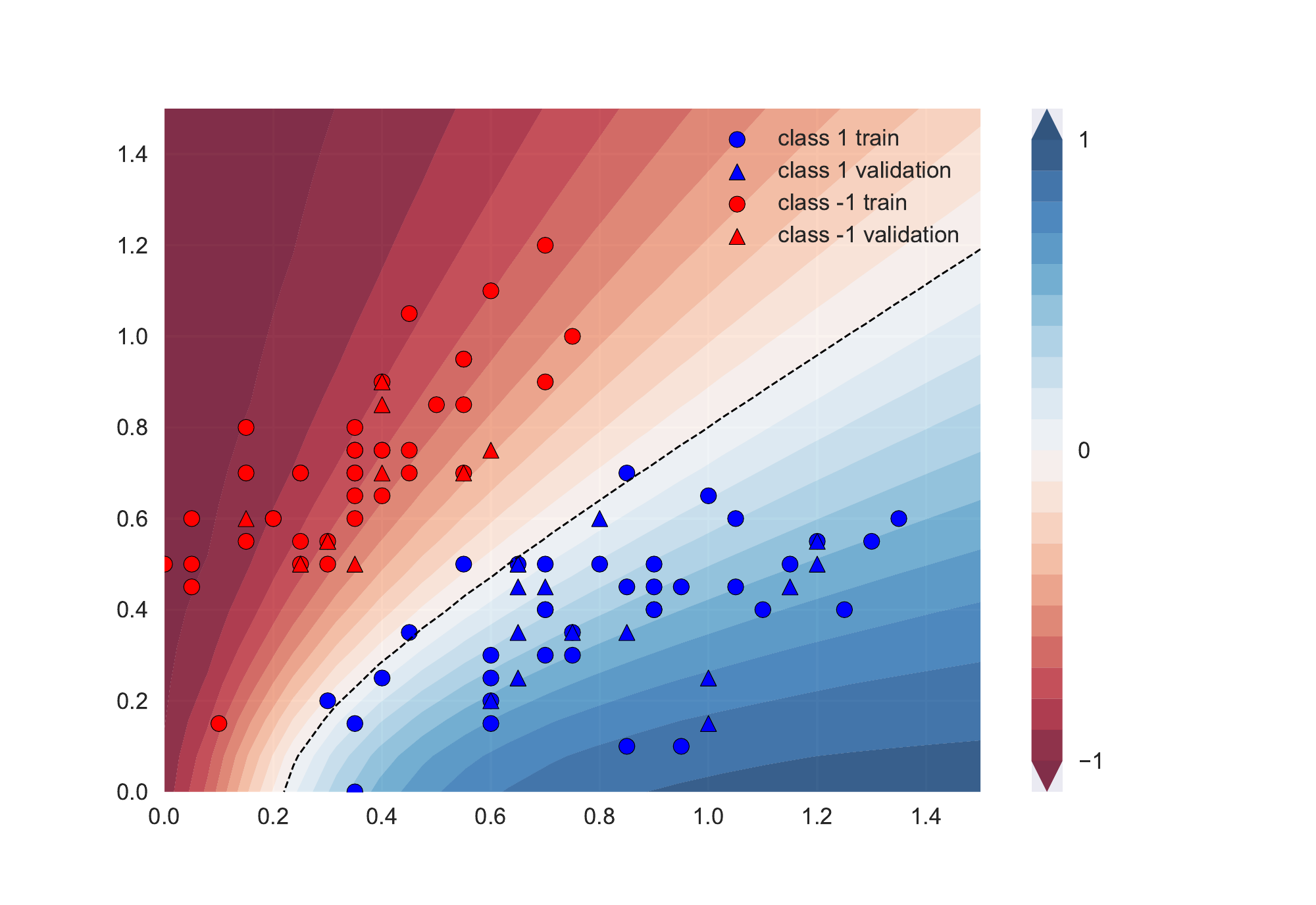}\label{fig:va:c}}
\subfigure[WS-SGD Warm-start with Multiple Random Sample]{\includegraphics[width=.5\columnwidth]{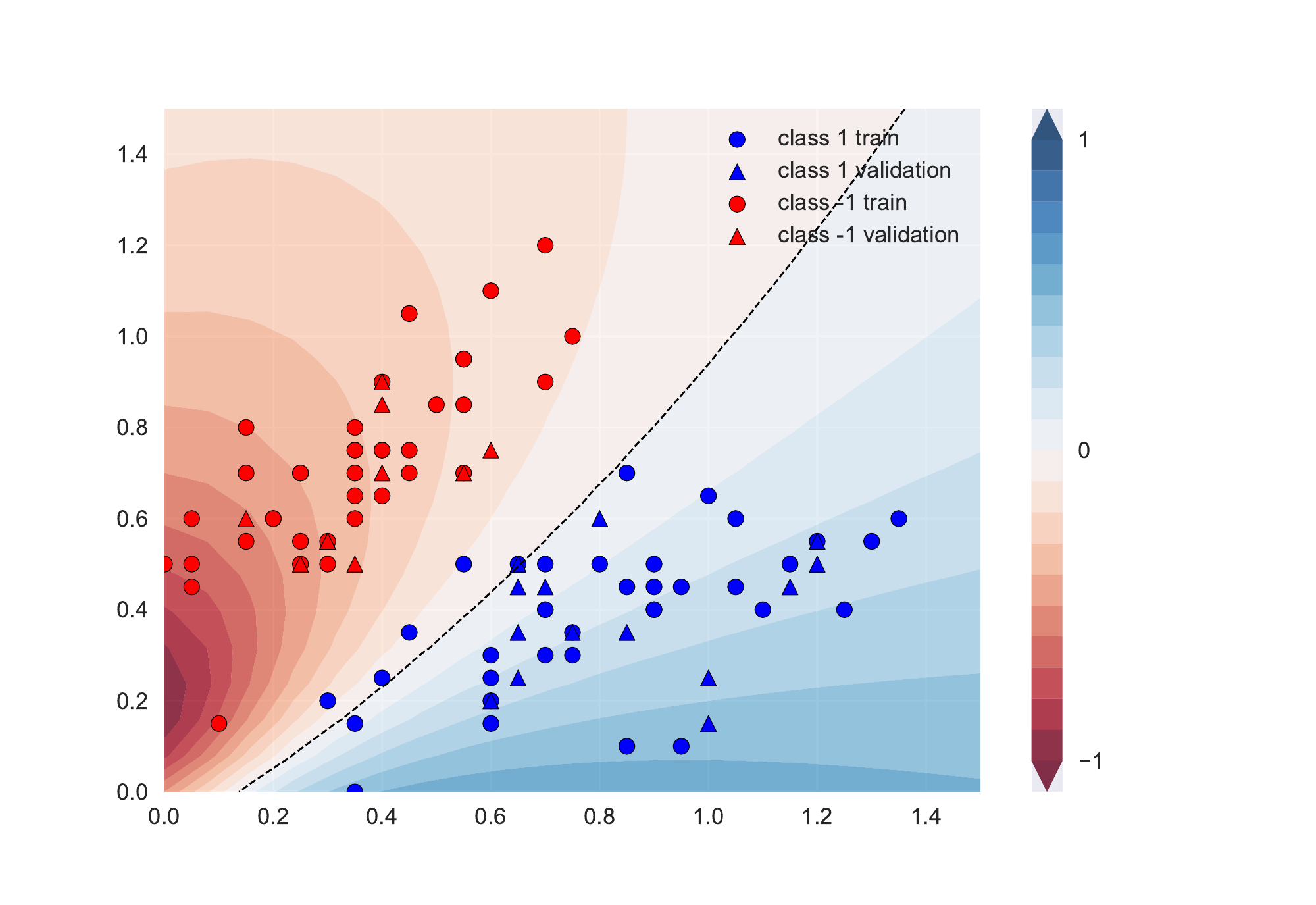}\label{fig:va:d}}
\caption{Variational classifier for Iris classification task: Variants of LAWS}
\label{fig:variants}
\end{figure*}

% \begin{figure*}[t]
% \centering
% \subfigure[]{\includegraphics[width=.8\columnwidth]{figs/cost.pdf}\label{fig:class:a}}
% \subfigure[]{\includegraphics[width=.8\columnwidth]{figs/acc.pdf}\label{fig:class:b}}
% \subfigure[]{\includegraphics[width=.8\columnwidth]{figs/class_sgd.pdf}\label{fig:class:c}}
% \subfigure[]{\includegraphics[width=.8\columnwidth]{figs/class_ws.pdf}\label{fig:class:d}}
% \caption{Variational classifier for Iris classification task}
% \label{fig:class}
% \end{figure*}
To evaluate the performance of LAWS, WS-SGD, and their variants, we use the open-source library PennyLane~\cite{bergholm2018pennylane} 0.22.2 built on Python 3.7.
Most of the experiments follow the open-source tutorials from official PennyLane website.
We first show that the LAWS could mitigate the BP issue with random PQC energy problems (in Figure~\ref{fig:randomPQC}).
Then we evaluate LAWS in a more practical problem of estimating the ground state energy of the hydrogen molecule  (in Figure~\ref{fig:HM}).
Further, we conduct experiments on variational quantum classifiers (in Figure~\ref{fig:class} and Figure~\ref{fig:variants}) such that the quantum circuits can be trained from labeled data to classify new data samples.
The classification training data is public and can be downloaded from the PennyLane tutorial.
All the experimental results and source code implementation can be found at \url{https://github.com/taozeyi1990/LAWS}.

\subsection{Evaluation on random PQC}
Figure~\ref{fig:randomPQC} above demonstrates the performance of LAWS compared to Vanilla SGD and Vanilla QNG in randomly designed PQC.
The random PQC is shown on (a), where it contains 3 qubits and 4 trainable parameters.
We perform the optimization for a total of 400 iterations.
It is undeniable that LAWS outperforms Vanilia QNG and Vanilla SGD, as shown on (b).
We divide the cost value region into 3 parts marked as blue~\ding{182}, red~\ding{183}, and green~\ding{184}.
In region~\ding{182}, all methods are quickly converging, and LAWS is even faster.
In region~\ding{183}, methods are trapped because of BP, and there is no further cost value decreasing.
Finally, in region~\ding{184}, LAWS and its variant mitigate the BP and find the optimal solutions.

Figure~\ref{fig:randomPQC} (c) and (d) are showing the different configurations of LAWS for different $\Delta$'s and look around step size $\eta$'s.
We see that the LAWS is sensitive in the change of $\eta$.

\subsection{Evaluation on Quantum Chemistry}
Figure~\ref{fig:HM} shows the experimental results for a more practical problem of Hydrogen VQE.
The primary purpose of this experiment is to approach ground state energy as close as we can.
The exact value of the ground state energy defined by the above PQC is given by $-1.1361894$.
The qubit register has been initialized to $|1100\rangle$, which encodes for the Hartree-Fock state of the hydrogen molecule described on a minimal basis.
We see the LAWS still outperforms than other state-of-the-art optimization methods.

\subsection{QNNs: Variational Classifier}
Last but least, we perform the binary Iris classification task, which is a simple but powerful QNN to show that the warm-start strategy has better generalization ability.
The learning rate for SGD and Nesterov momentum optimizer is set to be 0.01.
While the learning rate, look around rate and look around steps are 0.01, 0.5, and 5, respectively.
We train QNN model within 50 iterations.
Figure~\ref{fig:class:a},~\ref{fig:class:b}, and~\ref{fig:class:c}
show the cost function value, training accuracy, and validation accuracy, respectively.
As shown in each figure, the warm-start SGD in green demonstrates its superiority in this task.
Figure~\ref{fig:class:d},~\ref{fig:class:e}, and~\ref{fig:class:f} indicate the decision boundaries of model trained with different optimizers.
We observe that the two classes in the train and validation dataset are perfectly being separated when using the warm-start optimizer, which indicates the WSSGD has stronger generalization ability.
The result of the Nesterov optimizer seems to suffer from the under-fitted where the samples at the bottom left are mixed.

Moreover, we evaluate the different warm-start strategies discussed in section~\ref{subsec:implementation} of WSSGD in Figure~\ref{fig:variants}.
We test the impacts on WSSGD of different configurations such as sampling batch size.

\section{Conclusion}
In this work, we propose an unified framework for QNG by using a classical first-order optimization scheme.
The proposed new algorithm named WSSGD show its power in QVA learning. 
Our experiment results show that the proposed algorithm could mitigate the BP issue and have better generalization ability in quantum classification problems.

\bibliographystyle{IEEEtran}
\bibliography{qce22}

% Generated by IEEEtran.bst, version: 1.14 (2015/08/26)
\begin{thebibliography}{10}
\providecommand{\url}[1]{#1}
\csname url@samestyle\endcsname
\providecommand{\newblock}{\relax}
\providecommand{\bibinfo}[2]{#2}
\providecommand{\BIBentrySTDinterwordspacing}{\spaceskip=0pt\relax}
\providecommand{\BIBentryALTinterwordstretchfactor}{4}
\providecommand{\BIBentryALTinterwordspacing}{\spaceskip=\fontdimen2\font plus
\BIBentryALTinterwordstretchfactor\fontdimen3\font minus
  \fontdimen4\font\relax}
\providecommand{\BIBforeignlanguage}[2]{{%
\expandafter\ifx\csname l@#1\endcsname\relax
\typeout{** WARNING: IEEEtran.bst: No hyphenation pattern has been}%
\typeout{** loaded for the language `#1'. Using the pattern for}%
\typeout{** the default language instead.}%
\else
\language=\csname l@#1\endcsname
\fi
#2}}
\providecommand{\BIBdecl}{\relax}
\BIBdecl

\bibitem{preskill2018quantum}
J.~Preskill, ``Quantum computing in the nisq era and beyond,'' \emph{Quantum},
  vol.~2, p.~79, 2018.

\bibitem{shor1994algorithms}
P.~W. Shor, ``Algorithms for quantum computation: discrete logarithms and
  factoring,'' in \emph{Proceedings 35th annual symposium on foundations of
  computer science}.\hskip 1em plus 0.5em minus 0.4em\relax Ieee, 1994, pp.
  124--134.

\bibitem{lloyd1996universal}
S.~Lloyd, ``Universal quantum simulators,'' \emph{Science}, vol. 273, no. 5278,
  pp. 1073--1078, 1996.

\bibitem{harrow2009quantum}
A.~W. Harrow, A.~Hassidim, and S.~Lloyd, ``Quantum algorithm for linear systems
  of equations,'' \emph{Physical review letters}, vol. 103, no.~15, p. 150502,
  2009.

\bibitem{arute2019quantum}
F.~Arute, K.~Arya, R.~Babbush, D.~Bacon, J.~C. Bardin, R.~Barends, R.~Biswas,
  S.~Boixo, F.~G. Brandao, D.~A. Buell \emph{et~al.}, ``Quantum supremacy using
  a programmable superconducting processor,'' \emph{Nature}, vol. 574, no.
  7779, pp. 505--510, 2019.

\bibitem{cerezo2021variational}
M.~Cerezo, A.~Arrasmith, R.~Babbush, S.~C. Benjamin, S.~Endo, K.~Fujii, J.~R.
  McClean, K.~Mitarai, X.~Yuan, L.~Cincio \emph{et~al.}, ``Variational quantum
  algorithms,'' \emph{Nature Reviews Physics}, vol.~3, no.~9, pp. 625--644,
  2021.

\bibitem{mcclean2016theory}
J.~R. McClean, J.~Romero, R.~Babbush, and A.~Aspuru-Guzik, ``The theory of
  variational hybrid quantum-classical algorithms,'' \emph{New Journal of
  Physics}, vol.~18, no.~2, p. 023023, 2016.

\bibitem{peruzzo2014variational}
A.~Peruzzo, J.~McClean, P.~Shadbolt, M.-H. Yung, X.-Q. Zhou, P.~J. Love,
  A.~Aspuru-Guzik, and J.~L. O’brien, ``A variational eigenvalue solver on a
  photonic quantum processor,'' \emph{Nature communications}, vol.~5, no.~1,
  pp. 1--7, 2014.

\bibitem{farhi2014quantum}
E.~Farhi, J.~Goldstone, and S.~Gutmann, ``A quantum approximate optimization
  algorithm,'' \emph{arXiv preprint arXiv:1411.4028}, 2014.

\bibitem{farhi2018classification}
E.~Farhi and H.~Neven, ``Classification with quantum neural networks on near
  term processors,'' \emph{arXiv preprint arXiv:1802.06002}, 2018.

\bibitem{cerezo2021cost}
M.~Cerezo, A.~Sone, T.~Volkoff, L.~Cincio, and P.~J. Coles, ``Cost function
  dependent barren plateaus in shallow parametrized quantum circuits,''
  \emph{Nature communications}, vol.~12, no.~1, pp. 1--12, 2021.

\bibitem{tilly2021variational}
J.~Tilly, H.~Chen, S.~Cao, D.~Picozzi, K.~Setia, Y.~Li, E.~Grant, L.~Wossnig,
  I.~Rungger, G.~H. Booth \emph{et~al.}, ``The variational quantum eigensolver:
  a review of methods and best practices,'' \emph{arXiv preprint
  arXiv:2111.05176}, 2021.

\bibitem{bittel2021training}
L.~Bittel and M.~Kliesch, ``Training variational quantum algorithms is
  np-hard,'' \emph{Physical Review Letters}, vol. 127, no.~12, p. 120502, 2021.

\bibitem{stokes2020quantum}
J.~Stokes, J.~Izaac, N.~Killoran, and G.~Carleo, ``Quantum natural gradient,''
  \emph{Quantum}, vol.~4, p. 269, 2020.

\bibitem{yamamoto2019natural}
N.~Yamamoto, ``On the natural gradient for variational quantum eigensolver,''
  \emph{arXiv preprint arXiv:1909.05074}, 2019.

\bibitem{koczor2019quantum}
B.~Koczor and S.~C. Benjamin, ``Quantum natural gradient generalised to
  non-unitary circuits,'' \emph{arXiv preprint arXiv:1912.08660}, 2019.

\bibitem{carleo2017solving}
G.~Carleo and M.~Troyer, ``Solving the quantum many-body problem with
  artificial neural networks,'' \emph{Science}, vol. 355, no. 6325, pp.
  602--606, 2017.

\bibitem{mcardle2018variational}
S.~McArdle, S.~Endo, Y.~Li, S.~Benjamin, and X.~Yuan, ``Variational quantum
  simulation of imaginary time evolution with applications in chemistry and
  beyond,'' \emph{arXiv preprint arXiv:1804.03023}, 2018.

\bibitem{mcclean2018barren}
J.~R. McClean, S.~Boixo, V.~N. Smelyanskiy, R.~Babbush, and H.~Neven, ``Barren
  plateaus in quantum neural network training landscapes,'' \emph{Nature
  communications}, vol.~9, no.~1, pp. 1--6, 2018.

\bibitem{suzuki2021normalized}
Y.~Suzuki, H.~Yano, R.~Raymond, and N.~Yamamoto, ``Normalized gradient descent
  for variational quantum algorithms,'' in \emph{2021 IEEE International
  Conference on Quantum Computing and Engineering (QCE)}.\hskip 1em plus 0.5em
  minus 0.4em\relax IEEE, 2021, pp. 1--9.

\bibitem{haug2021optimal}
T.~Haug and M.~Kim, ``Optimal training of variational quantum algorithms
  without barren plateaus,'' \emph{arXiv preprint arXiv:2104.14543}, 2021.

\bibitem{grant2019initialization}
E.~Grant, L.~Wossnig, M.~Ostaszewski, and M.~Benedetti, ``An initialization
  strategy for addressing barren plateaus in parametrized quantum circuits,''
  \emph{Quantum}, vol.~3, p. 214, 2019.

\bibitem{liu2021parameter}
H.-Y. Liu, T.-P. Sun, Y.-C. Wu, Y.-J. Han, and G.-P. Guo, ``A parameter
  initialization method for variational quantum algorithms to mitigate barren
  plateaus based on transfer learning,'' \emph{arXiv preprint
  arXiv:2112.10952}, 2021.

\bibitem{franken2020gradient}
L.~Franken, B.~Georgiev, S.~Muecke, M.~Wolter, N.~Piatkowski, and C.~Bauckhage,
  ``Gradient-free quantum optimization on nisq devices,'' \emph{arXiv preprint
  arXiv:2012.13453}, 2020.

\bibitem{Nemirovskii1983}
A.~S. Nemirovskii, D.~B. Yudin, D.~B. Iudin, and D.~B. Iudin, \emph{Problem
  complexity and method efficiency in optimization}.\hskip 1em plus 0.5em minus
  0.4em\relax Wiley, 1983.

\bibitem{wierichs2020avoiding}
D.~Wierichs, C.~Gogolin, and M.~Kastoryano, ``Avoiding local minima in
  variational quantum eigensolvers with the natural gradient optimizer,''
  \emph{Physical Review Research}, vol.~2, no.~4, p. 043246, 2020.

\bibitem{cerezo2020impact}
M.~Cerezo and P.~J. Coles, ``Impact of barren plateaus on the hessian and
  higher order derivatives,'' \emph{arXiv e-prints}, pp. arXiv--2008, 2020.

\bibitem{arrasmith2021effect}
A.~Arrasmith, M.~Cerezo, P.~Czarnik, L.~Cincio, and P.~J. Coles, ``Effect of
  barren plateaus on gradient-free optimization,'' \emph{Quantum}, vol.~5, p.
  558, 2021.

\bibitem{sutskever2013importance}
I.~Sutskever, J.~Martens, G.~Dahl, and G.~Hinton, ``On the importance of
  initialization and momentum in deep learning,'' in \emph{International
  conference on machine learning}.\hskip 1em plus 0.5em minus 0.4em\relax PMLR,
  2013, pp. 1139--1147.

\bibitem{harrow2021low}
A.~W. Harrow and J.~C. Napp, ``Low-depth gradient measurements can improve
  convergence in variational hybrid quantum-classical algorithms,''
  \emph{Physical Review Letters}, vol. 126, no.~14, p. 140502, 2021.

\bibitem{duchi2011adaptive}
J.~Duchi, E.~Hazan, and Y.~Singer, ``Adaptive subgradient methods for online
  learning and stochastic optimization.'' \emph{Journal of machine learning
  research}, vol.~12, no.~7, 2011.

\bibitem{kingma2014adam}
D.~P. Kingma and J.~Ba, ``Adam: A method for stochastic optimization,''
  \emph{arXiv preprint arXiv:1412.6980}, 2014.

\bibitem{zhuang2020adabelief}
J.~Zhuang, T.~Tang, Y.~Ding, S.~C. Tatikonda, N.~Dvornek, X.~Papademetris, and
  J.~Duncan, ``Adabelief optimizer: Adapting stepsizes by the belief in
  observed gradients,'' \emph{Advances in neural information processing
  systems}, vol.~33, pp. 18\,795--18\,806, 2020.

\bibitem{amari1998natural}
S.-I. Amari, ``Natural gradient works efficiently in learning,'' \emph{Neural
  computation}, vol.~10, no.~2, pp. 251--276, 1998.

\bibitem{jones2020efficient}
T.~Jones and J.~Gacon, ``Efficient calculation of gradients in classical
  simulations of variational quantum algorithms,'' \emph{arXiv preprint
  arXiv:2009.02823}, 2020.

\bibitem{soori2021tengrad}
S.~Soori, B.~Can, B.~Mu, M.~G{\"u}rb{\"u}zbalaban, and M.~M. Dehnavi,
  ``Tengrad: Time-efficient natural gradient descent with exact fisher-block
  inversion,'' \emph{arXiv preprint arXiv:2106.03947}, 2021.

\bibitem{nelder1965simplex}
J.~A. Nelder and R.~Mead, ``A simplex method for function minimization,''
  \emph{The computer journal}, vol.~7, no.~4, pp. 308--313, 1965.

\bibitem{powell1964efficient}
M.~J. Powell, ``An efficient method for finding the minimum of a function of
  several variables without calculating derivatives,'' \emph{The computer
  journal}, vol.~7, no.~2, pp. 155--162, 1964.

\bibitem{holmes2022connecting}
Z.~Holmes, K.~Sharma, M.~Cerezo, and P.~J. Coles, ``Connecting ansatz
  expressibility to gradient magnitudes and barren plateaus,'' \emph{PRX
  Quantum}, vol.~3, no.~1, p. 010313, 2022.

\bibitem{larocca2021diagnosing}
M.~Larocca, P.~Czarnik, K.~Sharma, G.~Muraleedharan, P.~J. Coles, and
  M.~Cerezo, ``Diagnosing barren plateaus with tools from quantum optimal
  control,'' \emph{arXiv preprint arXiv:2105.14377}, 2021.

\bibitem{wecker2015progress}
D.~Wecker, M.~B. Hastings, and M.~Troyer, ``Progress towards practical quantum
  variational algorithms,'' \emph{Physical Review A}, vol.~92, no.~4, p.
  042303, 2015.

\bibitem{van2021measurement}
B.~van Straaten and B.~Koczor, ``Measurement cost of metric-aware variational
  quantum algorithms,'' \emph{PRX Quantum}, vol.~2, no.~3, p. 030324, 2021.

\bibitem{yuan2019theory}
X.~Yuan, S.~Endo, Q.~Zhao, Y.~Li, and S.~C. Benjamin, ``Theory of variational
  quantum simulation,'' \emph{Quantum}, vol.~3, p. 191, 2019.

\bibitem{arrasmith2021equivalence}
A.~Arrasmith, Z.~Holmes, M.~Cerezo, and P.~J. Coles, ``Equivalence of quantum
  barren plateaus to cost concentration and narrow gorges,'' \emph{arXiv
  preprint arXiv:2104.05868}, 2021.

\bibitem{harrow2009random}
A.~W. Harrow and R.~A. Low, ``Random quantum circuits are approximate
  2-designs,'' \emph{Communications in Mathematical Physics}, vol. 291, no.~1,
  pp. 257--302, 2009.

\bibitem{nemirovski2009robust}
A.~Nemirovski, A.~Juditsky, G.~Lan, and A.~Shapiro, ``Robust stochastic
  approximation approach to stochastic programming,'' \emph{SIAM Journal on
  optimization}, vol.~19, no.~4, pp. 1574--1609, 2009.

\bibitem{bottou2018optimization}
L.~Bottou, F.~E. Curtis, and J.~Nocedal, ``Optimization methods for large-scale
  machine learning,'' \emph{Siam Review}, vol.~60, no.~2, pp. 223--311, 2018.

\bibitem{zhang2019lookahead}
M.~Zhang, J.~Lucas, J.~Ba, and G.~E. Hinton, ``Lookahead optimizer: k steps
  forward, 1 step back,'' \emph{Advances in Neural Information Processing
  Systems}, vol.~32, 2019.

\bibitem{zhou2021towards}
P.~Zhou, H.~Yan, X.~Yuan, J.~Feng, and S.~Yan, ``Towards understanding why
  lookahead generalizes better than sgd and beyond,'' \emph{Advances in Neural
  Information Processing Systems}, vol.~34, 2021.

\bibitem{bergholm2018pennylane}
V.~Bergholm, J.~Izaac, M.~Schuld, C.~Gogolin, M.~S. Alam, S.~Ahmed, J.~M.
  Arrazola, C.~Blank, A.~Delgado, S.~Jahangiri \emph{et~al.}, ``Pennylane:
  Automatic differentiation of hybrid quantum-classical computations,''
  \emph{arXiv preprint arXiv:1811.04968}, 2018.

\end{thebibliography}

\end{document}